\documentclass{vldb}
\usepackage{theorem,amsmath,amssymb,latexsym,xspace,ifthen,url}

\newtheorem{theorem}{Theorem}[section]
\newtheorem{proposition}[theorem]{Proposition}
\newtheorem{lemma}[theorem]{Lemma}
\newtheorem{corollary}[theorem]{Corollary}

\newcommand{\SortNoop}[1]{}

\newtheorem{defn}[theorem]{Definition}
\newenvironment{definition}[1][nonsense]{\ifthenelse{\equal{#1}{nonsense}}{\begin{defn}\upshape}{\begin{defn}[#1]\upshape}}{\end{defn}}

\newcommand{\PNPlogsq}{\textup{\rm P$^{\text{\rm NP[log$^2$]}}$}\xspace}
\newcommand{\PNP}{\textup{\rm P$^{\text{\rm NP}}$}\xspace}
\newcommand{\NP}{\textup{\rm NP}}
\newcommand{\coNP}{\textup{\rm coNP}}
\newcommand{\LFP}{\text{\rm LFP}\xspace}
\newcommand{\GFP}{\text{\rm GFP}\xspace}

\newcommand{\limp}{\rightarrow}

\newcommand{\atoms}{{\mathrm{atoms}}}

\newcommand{\tr}{\textsf{tr}}

\newcommand{\adom}{\mathit{adom}}

\newtheorem{claim}{Claim}

\title{Queries with Guarded Negation (full version)\titlenote{Detailed proofs of the results in this paper can be found in the appendix.
We would like to thank Alkis Polyzotis for helpful comments.} }

\numberofauthors{3}
\author{
\alignauthor {Vince B\'ar\'any} \\
   \affaddr{TU Darmstadt, Germany}\\
       \email{\texttt{barany@mathematik.tu-darmstadt.de}}
\alignauthor {Balder ten Cate} \\
   \affaddr{UC Santa Cruz, CA, USA}\\
       \email{\texttt{btencate@ucsc.edu}}
\alignauthor {Martin Otto} \\
   \affaddr{TU Darmstadt, Germany}\\
       \email{\texttt{otto@mathematik.tu-darmstadt.de}}  
   }
\begin{document}
\maketitle

\begin{abstract}
A well-established and fundamental insight in database theory is that negation (also known as complementation) tends to make queries difficult to process and difficult to reason about. Many basic problems are decidable and admit practical algorithms in the case of unions of conjunctive queries, but become difficult or even undecidable when queries are allowed to contain negation. 
Inspired by recent results in finite model theory, we consider a restricted form of negation, \emph{guarded negation}. We introduce a fragment of SQL, called GN-SQL, as well as a fragment of Datalog with stratified negation, called GN-Datalog, that allow only guarded negation, and we show that these query languages are computationally well behaved, in terms of testing query containment, query evaluation, open-world query answering, and boundedness.
GN-SQL and GN-Datalog subsume a number of well known query languages and constraint languages, such as unions of conjunctive queries, monadic Datalog, and frontier-guarded tgds. In addition, an analysis of standard benchmark workloads shows that most usage of negation in SQL in practice is  guarded negation.
\end{abstract}

\section{Introduction}

A well-established and fundamental insight of database theory is that \emph{negation} (also called \emph{complementation} or \emph{difference}) tends to make queries difficult to reason about. Recall that the unions of conjunctive queries are the first-order queries that can be expressed without using negation. Many basic problems are decidable and admit practical algorithms in the case of unions of conjunctive queries, but are undecidable in the case of arbitrary first-order queries. Examples include \emph{query containment} and \emph{open world query answering}. 

We argue that most queries in practice use only a restricted form of
negation, which is called \emph{guarded negation} and was first
considered in~\cite{BtCS11} (in the study of decidable fragments of
first-order logic).  By guarded negation we mean that queries may involve negative conditions only if these conditions, intuitively, pertain to a single record in the database.  For instance, if a database schema contains relations Author(AuthID,Name) and Book(AuthID,Title,Year,Publisher), the query that asks for \emph{authors that did not publish any book with Elsevier}
since ``not publishing a book with Elsevier'' is a property of an author. The query that asks for \emph{pairs of authors and book titles where the author did not publish the book}, on the other hand, is not allowed,
since it involves a negative condition (in this case, an inequality) pertaining to two values that do not necessarily co-occur in a record in the database.  
The requirement of guarded negation can be formally stated most easily in terms of the Relational Algebra: we allow the use of the difference operator $E_1-E_2$ provided that $E_1$ is a projection of a relation from the database.

Based on an analysis of standard SQL benchmark workloads, we argue
that guarded negation covers most uses of negation in SQL in practice.
Furthermore, building on recent results in logic and finite model
theory~\cite{BtCS11,BOW}, we show that queries with guarded negation
are computationally very well behaved. For instance, query containment
and open world query answering are decidable for first-order queries
with guarded negation, and boundedness is decidable for the
guarded-negation fragment of Datalog with stratified negation. We also
determine the complexity of query evaluation for queries with guarded negation, which (under reasonable complexity theoretic assumptions) is easier than the same problem for queries with unguarded negation.

Our results show that guarded negation is a fruitful concept for databases, in the sense that it enables solving central decision problems in database theory more efficiently.  We also believe that guarded negation is a fruitful concept from a more practical point of view, allowing for efficient query plans and query optimization strategies. This is something we are exploring in a separate line of investigation.

\medskip\par\noindent\textbf{Outline and Main Results.}
In Section~\ref{sec:prel} we review the definition of GNFO, guarded-negation first-order logic, and GNFP, guarded-negation fixed-point logic, as well as the main known decidability and complexity results for these logics \cite{BtCS11}. We also provide an equivalent characterization of GNFO in terms of the Relational Algebra.

In Section~\ref{sec:GNSQL}, we investigate what it means for an SQL query to be negation-guarded. Specifically, we identify syntactic restrictions on the use of negation in SQL queries, and we show that the first-order queries satisfying  these restrictions  can be translated to GNFO, and, in fact, are expressively complete for GNFO, in the sense of Codd's completeness theorem. Furthermore, by means of an analysis of standard SQL benchmark workloads, we show that most SQL queries in practice satisfy the syntactic restrictions. 

In Section~\ref{GNDsec}, similarly, we introduce a syntactic fragment
of Datalog with stratified negation, called GN-Datalog, which admits a
translation into GNFP.

In Section~\ref{sec:existing}, we show that GN-SQL and GN-Datalog subsume a number of important existing query languages and constraint languages. In particular, GN-Datalog subsumes both monadic Datalog (which it extends by allowing IDBs of arbitrary arity, and negation, subject to guardedness conditions) and unions of conjunctive queries.

In Section~\ref{sec:containment}, we show that query containment 
is 2ExpTime-complete for GN-SQL queries as well as for 
GN-Datalog queries (note that the decidability of these problems 
follows via translations into GNFO and GNFP). %

In Section~\ref{sec:answering}, we determine the complexity of query evaluation and open-world query answering for GN-SQL and for GN-Datalog. While the data complexity of query evaluation is in PTime, both for GN-SQL and for GN-Datalog, in terms of combined complexity, the problem is complete for the complexity class \PNPlogsq (for GN-SQL) and \PNP (for GN-Datalog).  
The data complexity of \emph{open world query answering} for GN-SQL 
with respect to incomplete databases is coNP-complete. 
The problem can be solved in PTime for a considerable fragment of GN-SQL.

In Section~\ref{sec:boundedness}, we prove decidability of the \emph{boundedness} problem 
for GN-Datalog. Boundedness is a classical decision problem in the study of query 
optimization for recursive queries. It is known to be undecidable for Datalog, 
but decidable for monadic Datalog. 
Our result can be viewed as a powerful generalization of the decidability 
of boundedness for monadic Datalog queries \cite{CosmadakisGaKaVa88}.

We conclude in Section~\ref{sec:discussion} by discussing possibilities for further extending GN-SQL and GN-Datalog.

\section{Preliminaries}\label{sec:prel}

In this section, we review definitions and results concerning the guarded-negation logics GNFO and GNFP from \cite{BtCS11}. These results will be put to extensive use in the rest of this paper. We assume familiarity with the basic syntax and semantics of first-order logic. 

For clarity, we will maintain a distinction between instances and structures: a structure has an associated domain, which may be a superset of its active domain, and which may depend on the structure in question. Furthermore, structures may interpret not only relation symbols but also constant symbols (which denote, not necessarily distinct, domain elements). 
Thus, instances may be viewed as a special case of structures, where the domain is the active domain and there are no constant symbols. 
Unless explicitly stated otherwise (by means of the adjective ``unrestricted''), we always assume structures and instances to be finite.

\vskip -.4em
\paragraph*{GNFO} Guarded Negation First-Order Logic (GNFO) is the fragment of first-order logic consisting of all formulas built up from atomic formulas (including equalities) using
  conjunction, disjunction, existential quantification, and guarded negation, that is, negation in the specific form $\alpha\land\neg\phi$ where $\alpha$
  is an atomic formula (possibly an equality statement) and all free variables of $\phi$ occur in $\alpha$. Note that, since the guard $\alpha$ is
  allowed to be an equality statement, we are essentially able to negate any formula with at most one free variable 
  (by writing $x=x\land\neg\phi(x)$).  Formally, the formulas of GNFO are generated by the recursive definition
\[ \phi ::= R(t_1, \ldots, t_n) \mid t_1=t_2 \mid \phi_1\land\phi_2 \mid \phi_1\lor\phi_2 \mid \exists x\phi\mid \alpha\land\neg\phi \]
where each $t_i$ is either a variable or a constant symbol, and, in the last clause, $\alpha$ is an atomic formula containing all free variables of $\phi$.
Note that function symbols (of arity greater than zero) are not considered.

  In the above definition, we required $\alpha$ to be an atomic formula containing all free variables of the negated formula $\phi$. Occasionally, it is convenient to allow a slightly more liberal syntax.
 Let us say that 
 $\alpha$ is a \emph{generalized guard} for $\phi$ if $\alpha$ is a disjunction of existentially quantified atomic formulas such that the free variables of $\phi$ are included in the free variables of each disjunct. One could extend GNFO by allowing generalized guards in the definition of guarded negation, thus admitting formulas such as $(\exists uv\,R(x,y,u,v) \lor \exists uv\,R(y,x,u,v))\land\neg Sxy$.  This would not increase the expressive power of GNFO: if a negation is guarded by a generalized guard, we can ``pull out'' the disjunction and the existential quantification  to obtain an equivalent formula without generalized guards (at the cost of a possibly exponential blow-up in formula size). In particular, the above example can be equivalently expressed by
$\exists uv(R(x,y,u,v)\land \neg Sxy)\lor \exists uv(R(y,x,u,v)\land \neg Sxy)$. Therefore, for simplicity, our definition of GNFO does not allow for generalized guards. %

\vskip -.4em
\paragraph*{GNFP} 
Guarded Negation Fixed Point Logic (GNFP) further extends GNFO with an operator 
for least fixed points of positively definable monotone operations on relations. 
That is, we introduce second-order variables (also called fixed-point variables) 
of arbitrary arity, which may be used to form atomic formulas in the same way 
as ordinary relation symbols, and if $\phi$ is any GNFP formula, $X$ an $n$-ary 
second-order variable $(n\geq 1)$ occurring only positively in~$\phi$ (i.e, 
under an even number of negations), $\textbf{x}=x_1, \ldots, x_n$ a sequence of 
first-order variables, and $\textbf{t}=t_1, \ldots, t_n$ a sequence of terms 
(first-order variables or constant symbols), and the free first-order variables 
of $\phi$ are included in $\textbf{x}$, then 
\begin{center}
$[\LFP_{X,\textbf{x}}~\alpha\land\phi](\textbf{t})$ 
\end{center}
is also a formula of GNFP, where $\alpha$ is a generalized guard for $\phi$, 
i.e., a disjunction of existentially quantified atomic formulas (involving only 
atomic relations, no second-order variables), such that all free first-order 
variables of $\phi$ are also free variables of each disjunct of $\alpha$.%

In the above formula, the \LFP operator is a generalized quantifier 
binding the variables $X$ and $\textbf{x}$.
The formula expresses that the tuple $\textbf{t}$ belongs to the least 
fixed-point of the monotone operation on relations defined by $\alpha\land\phi$.
Incidentally, here, unlike in the case of GNFO, it is important that $\alpha$ 
is allowed to be a \emph{generalized} guard. 

In what follows, whenever we consider \LFP formulas, we will always assume 
that they do not have any free second-order variables. 
The formal semantics of $[\LFP_{X,\textbf{x}}~\alpha\land\phi](\textbf{y})$ 
is the familiar one from least fixed-point logic, cf.~\cite{AHV95}.  
If the formula $\phi$ has at most one free variable $x$, we may omit 
the guard $\alpha$, which can be assumed to be the equality statement $x=x$. 
For example, the GNFP formula
\[
   [\LFP_{X,x} ~ P(x) \lor \exists y~R(x,y)\land X(y)](z)
\]
says that there is an $R$-path from $z$ to some element in $P$. 

\vskip -.4em
\paragraph*{Definability of Greatest Fixed Points}
Besides the above \emph{least fixed-point} operator, we can consider an analogous operator \GFP for taking the \emph{greatest fixed point} of a definable monotone operation on relations. 
However, as it turns out, it is possible to define the \GFP operator in terms of the \LFP operator (and vice versa) using a dualization via guarded negation. Specifically, $[\GFP_{X,\textbf{x}}~\alpha(\textbf{x})\land\phi(\textbf{x})](\textbf{t})$ can be equivalently expressed as
$\alpha(\textbf{t})\land\neg[\LFP_{X,\textbf{x}} \alpha(\textbf{x})\land\neg\phi'(\textbf{x})](\textbf{t})$, where $\phi'$ is obtained from $\phi$ by replacing all subformulas of the form 
$X(\textbf{t}')$ by $\alpha(\textbf{t}')\land\neg X(\textbf{t}')$.  
For this reason, the above definition of GNFP does not include \GFP as a primitive operator.

\vskip -.4em
\paragraph*{Definability of Simultaneous Fixed Points}
It is common, in the literature on fixed point logics, to consider also a \emph{simultaneous} least fixed point operator, that takes as arguments not a single formula but a tuple of formulas. More precisely, in the context of GNFP it is natural to consider also formulas of the form $[\LFP_{X_i} S](\textbf{t})$ where 
\[\small 
  S = \begin{cases}
         X_1(\textbf{x}_1) &\!\!\leftarrow~ 
            \alpha_1(\textbf{x}_1) \land \phi_1(X_1, \ldots, X_n,\textbf{x}_1) \\ 
        &\tiny{\vdots} \\ 
        X_n(\textbf{x}_n) &\!\!\leftarrow~ 
            \alpha_n(\textbf{x}_n)\land \phi_n(X_1, \ldots, X_n, \textbf{x}_n)
       \end{cases}
\]
is a system of GNFP formulas, with each $X_k$ a distinct second-order variable, whose arity matches the length of the tuple $\textbf{x}_k$, and which occurs only positively in $\phi_1, \ldots, \phi_n$, and where $\textbf{t}$ is a tuple of terms of the same length as $\textbf{x}_i$. Here, the system $S$ can be viewed as defining a monotone operation
on tuples of relations, and the above formula expresses that $\textbf{t}$ belongs to the $i$-th component of the least fixed point of this operation.
It is well known that simultaneous fixed point expressions of this form can be expressed equivalently using a nesting of ordinary, single-variable, fixed point operators, possibly at the cost of an exponential blow-up in formula size (cf.~for example~\cite{ArnoldNiw}). Hence,
extending GNFP with such a simultaneous least fixed point operator does not increase its expressive power. 

\vskip -.4em
\paragraph*{Disjunctive Normal Form for GNFO and Width}
 We say that a GNFO formula is in \emph{Disjunctive Normal Form} (DNF) 
if it is a disjunction of disjunction-free GNFO formulas, 
no existential quantifier occurs directly below a conjunction sign, 
and no conjunction sign occurs directly below a negation sign. 
Equivalently, a GNFO formula is in DNF if it is a disjunction of 
GNFO formulas $\phi$ generated by the following recursive definition:
\begin{equation} \label{eq:DNF}
\begin{array}{lll} 
   \phi &::=& \exists x_1, \ldots, x_n(\zeta_1\land\cdots\land\zeta_m) \\
  \zeta &::=& R(t_1, \ldots, t_n) ~\mid~ (t_1=t_2) \mid \alpha\land\neg\phi 
\end{array}
\end{equation}
where, in the last clause, $\alpha$ is an atomic formula containing 
all free variables of $\phi$.
Every GNFO formula is equivalent to one in DNF, of possibly exponential size, 
that can be obtained by repeatedly applying the following equivalences. 
\[\small\begin{array}{@{}l@{}}
  (\exists x\phi)\land\psi \simeq \exists x'(\phi[x'/x]\land\psi) , ~~\quad 
  \phi\land(\psi\lor\chi) \simeq (\phi\land\psi)\lor(\phi\land\chi) \\[.5em] 
  \exists x(\phi\lor\psi) \simeq \exists x\phi \lor \exists x\psi , ~~\quad\ 
  \alpha \land \lnot ( \phi \land \psi) \simeq (\alpha \land \lnot \phi) \lor (\alpha \land \lnot \psi) 
\end{array}\]
The \emph{width} of a GNFO formula $\phi$ is the number of variables occurring 
(free or bound) in the DNF formula obtained from $\phi$ by applying the above rules. 

A \emph{union of conjunctive queries} (UCQ) is a GNFO formula in DNF without negation.
Thus, GNFO can be naturally viewed as an extension of UCQs with guarded negation.

\vskip -.4em
\paragraph*{Known Decidability and Complexity Results}
The following theorem summarizes what is known about GNFO and GNFP that is relevant 
for present purposes. Recall that the satisfiability problem has as input 
a formula $\phi(\textbf{x})$, and asks whether there exists a structure $M$ 
and a tuple of elements $\textbf{a}$ such that $M\models\phi(\textbf{a})$. 
The entailment problem takes as input two formulas $\phi(\textbf{x})$, $\psi(\textbf{x})$, 
and asks whether it is the case that, for every structure $M$ and for every tuple 
of elements $\textbf{a}$, $M\models\phi(\textbf{a})$ implies $M\models\psi(\textbf{a})$. 
The model checking problem has as input a formula $\phi(\textbf{x})$, a structure $M$, 
and a tuple of elements $\textbf{a}$, and asks whether $M\models\phi(\textbf{a})$.

\begin{theorem}[\cite{BtCS11}]\label{thm:GNFO}
\begin{enumerate} \setlength{\itemsep}{-0.8mm}
\item The satisfiability problem and the entailment problem for GNFO and for GNFP 
      are decidable and 2ExpTime-complete. This holds both for finite structures 
      and for unrestricted structures.
\item For GNFO formulas, satisfiability over finite structure coincides with satisfiability 
      over unrestricted structures, and similarly for entailment. 
      The same does \emph{not} hold for GNFP.
\item The model checking problem for GNFO is \PNPlogsq-complete (combined complexity). 
      For GNFP, the problem is \PNP-hard and is contained in $\NP^\NP\cap \coNP^\NP$.
\end{enumerate}
\end{theorem}

In the above theorem, \PNPlogsq refers to those problems that can be solved by a 
polynomial time deterministic algorithm that is allowed to ask $O(\log^2(n))$ queries 
to an NP-oracle, cf.~Section~\ref{sec:evaluation}.
A close analysis of the 2ExpTime upper bound argument for the satisfiability and entailment problems of GNFP shows that these results extend to the case with simultaneous fixed-point operators (both on finite structures and on unrestricted structures).\footnote{
    \small Specifically, the proof of the 2ExpTime upper bound for GNFP is based on 
    a satisfiability preserving translation from GNFP to guarded fixed-point logic (GFP). 
    The translation may give rise to an exponential blow-up in the size of the formula, 
    but it preserves the width (following a suitable definition of width, analogous to 
    the definition of width for GNFO formulas). The satisfiability problem for GFP formulas, 
    in turn, is decidable in time $2^{poly(|\phi|)\cdot exp(\text{width}(\phi))}$ 
    (where $poly(n)$ is short for $n^{O(1)}$ and $exp(n)$ is short for $2^{poly(n)}$) 
    by a reduction to the emptiness problem for a suitable type of automata 
    \cite{GradelWalukiewicz99,BaranyBojanczyk11}. 
    The translation from GFP formulas to automata extends immediately to the case for 
    formulas containing simultaneous fixed-point operators.  
    Furthermore, the polynomial-time inductive satisfiability-preserving translation 
    from GNFP to GFP given in \cite{BtCS11} (which in fact simply commutes with the 
    fixed-point operators) extends in a straightforward manner to the case where the 
    input and output formulas may contain simultaneous fixed-point operators.}

\begin{figure*} 
\begin{center}
\begin{tabular}{lll}
\emph{query} &:=& \textsf{select} ($t_1 \textsf{ as } \textsc{attr}_1, \ldots, t_n \textsf{ as } \textsc{attr}_n$) \textsf{from} (\textsc{rel$_1$} $R_1$, \ldots, \textsc{rel$_m$} $R_m$) \textsf{where} \emph{condition}
\\ &&  $~\mid~$ \emph{query} \textsf{union} \emph{query} $~\mid~$ \emph{query} \textsf{intersect} \emph{query} $~\mid~$ \emph{query} \textsf{except} \emph{query} \\[1mm]
\emph{condition} &:=& \textsf{true} $~\mid~$ $t_1=t_2$ $~\mid~$ $t$ \textsf{in} \emph{query} $~\mid~$ \textsf{exists}(\emph{query})
  \\ && $~\mid~$ \emph{condition} \textsf{and} \emph{condition} $~\mid~$ \emph{condition} \textsf{or} \emph{condition} $~\mid~$ \textsf{not}(\emph{condition})
\end{tabular}
\end{center}
\vspace{-4mm}
\caption{Grammar for FO-SQL queries \label{fig:SQL}}
\end{figure*}

\subsection{Guarded Negation in Relational Algebra} \label{GNAlg}
The concept of guarded negation can be equivalently cast in terms of the Relational Algebra, where negation is expressed by means of the difference operator.  Consider the Relational Algebra (RA) defined over a schema consisting of relation symbols of specified arity using the following primitive operators (cf.~\cite{AHV95} for their semantics).
\begin{description} \setlength{\itemsep}{-0.9mm}
\item[Atomic Relations:] every relation symbol belongs to $\mathrm{RA}$.
\item[Selection:] if $E \in \mathrm{RA}$ has arity $k$ and $1\leq i,j\leq k$, 
   then $\sigma_{i=j}(E)$ belongs to $\mathrm{RA}$ and has arity $k$.
\item[Projection:] if $E\in \mathrm{RA}$ has arity $k$ and $1\leq i_1, \ldots, i_n\leq k$, 
   then $\pi_{i_1 \ldots i_n}(E)$ belongs to $\mathrm{RA}$ and has arity $n$.
\item[Crossproduct:] if $E_1,E_2\in \mathrm{RA}$ have arity $k$ and $n$, 
   respectively, then $E_1\times E_2$ belongs to $\mathrm{RA}$ and has arity $k+n$.
\item[Union, Intersection, and Difference:] if $E_1,E_2\in \mathrm{RA}$ 
   both have arity $k$, then $E_1\cup E_2$, $E_1\cap E_2$ and $E_1-E_2$ 
   belong to $\mathrm{RA}$ and have arity $k$.
\end{description}

Codd's completeness theorem states that $\mathrm{RA}$ has the same expressive power as the domain-independent fragment of first-order logic, cf.~\cite{AHV95}. 
Let us briefly recall here the definition of domain independence for
first-order formulas without constant symbols \cite{AHV95}. 
The \emph{active domain} of a structure $M$ is the set 
$\adom(M)$
of all elements that occur in a tuple belonging to one of the relations. For any structure $M$, let $M'$ be a copy of the same structure but where all elements outside $\adom(M)$ are removed.  A first-order formula $\phi(\textbf{x})$ without constant symbols is \emph{domain-independent} if (i) whenever $M\models\phi(\textbf{a})$, then the tuple $\textbf{a}$ consists of elements of $\adom(M)$, and (ii) for all tuples $\textbf{a}$ consisting of elements of $\adom(M)$, $M\models\phi(\textbf{a})$ if and only if $M'\models\phi(\textbf{a})$. The same definition applies to formulas with fixed-point operators.  Examples of first-order formulas that are \emph{not} domain-independent are $P(x)\lor Q(y)$, $x=x$, and $\neg P(x)$. 

We say that a relation algebra expression is \emph{negation-guarded} if every occurrence of the difference operator is of the form $\pi_{i_1 \ldots i_m}(R)-E$ where $R$ is a relation symbol. We denote by GN-RA the negation-guarded fragment of RA. It can be shown by straightforward inductive translations that GN-RA captures GNFO in the following sense.

\begin{theorem} \label{thm:GNRA}
  Every $k$-ary GN-RA expression is equivalent to a domain-independent 
  GNFO formula $\phi(x_1, \ldots, x_k)$, and vice versa, via a linear translation from GN-RA to GNFO and an exponential translation backwards.
\end{theorem}

Let $R, S$ be relation symbols of arity 2 and 1, respectively. 
The following RA expressions are \emph{not} negation-guarded.
\begin{trivlist}
\item $(\pi_1(R)\times S) - \pi_{1,1}(R)$ \hfill (distinct pairs from $\pi_1(R)\times S$)
\item  $\pi_{1,4}(\sigma_{2=3}(R\times R)) - R$ \hfill (reachability in two steps, not one)
\item $\pi_{1}(R) - \pi_{1}((\pi_{1}(R)\times S)-R)$  \hfill (the quotient $R\div S$)
\end{trivlist}
In fact, it follows from results in \cite{BtCS11} that none of these expressions is equivalent to a GN-RA expression.

Observe that in the above definition of GN-RA we did not allow for the use of constant values in selections and projections. This was only to simplify presentation. All complexity results that we will present go through in the presence of constant values, cf.~Section~\ref{sec:discussion}.

\section{Guarded Negation in SQL} \label{sec:GNSQL}

In this section, we discuss what it means for an SQL query to have guarded negation. More precisely, we consider a simple, first-order expressively complete, fragment of SQL with a set-based semantics, that we call FO-SQL, and we characterize the queries in this fragment that can be expressed in GNFO.

\vskip -.4em
\paragraph*{FO-SQL: a Simple First-Order Fragment of SQL} 
In this section, unlike in the rest of the paper, we work with named schemas. 
A \emph{named schema} is a collection of relation names, each with an associated list of attribute names. For the discussion below, assume we have a fixed schema, say, consisting of \textsc{book(isbn,author,title)} and \textsc{location(isbn,shelf,number)}. We also fix an infinite supply of ``tuple variables'' (also known as \emph{aliases}, and denoted by $R_1, R_2, \ldots$). 
By a \emph{term} $t$ we will mean an expression of the form $R_i.\textsc{attr}$ where $R_i$ is a tuple variable and \textsc{attr} is an attribute name. 

We consider SQL expressions that are generated by the simple grammar given in Figure~\ref{fig:SQL}, where each $t_i$ is a term, each \textsc{rel$_i$} is a relation name, $\textsc{attr}_1, \ldots, \textsc{attr}_n$ are distinct attribute names, and $R_1, \ldots, R_m$ are distinct tuple variables.  This grammar generates queries that may have free tuple variables, i.e., there may be occurrences of tuple-variables $R_i$ that are not in the scope of a \textsf{select-from-where} clause where they are declared. We will refer to queries with free tuple variables as \emph{open queries} (or \emph{correlated subqueries}), and we refer to queries without free tuple variables as \emph{closed queries} (or \emph{uncorrelated subqueries}). We will denote by $FV(q)$ the set of free tuple variables of $q$. We will be mainly interested in closed queries.  Note, however, that closed queries \emph{are} allowed to contain subexpressions of the form $\textsf{exists}(q)$ or of the form $t \textsf{ in } q$ where $q$ is an open query.

We only consider queries that are well-typed in the sense that each (open or closed) query 
can be consistently assigned a (unique) \emph{type}, which is a list of attribute names, where
\begin{enumerate} \setlength{\itemsep}{0em}
\item the type of a \textsf{select-from-where} query 
      is the set of attribute names specified in its \textsf{select} clause;  
\item the \textsf{union}, \textsf{intersect}, and \textsf{except} operators 
      take as arguments two queries of equal type, yielding a query of the same type. 
\end{enumerate}
Furthermore, terms $R_i.\textsc{attr}$ are only allowed to occur when $\textsc{attr}$ 
belongs to the schema of the relation to which the occurrence of $R_i$ in question is bound, 
and conditions of the form $t \textsf{ in } q$ are allowed only when $q$ is a unary query, 
i.e., when the type of $q$ consists of a single attribute.

By an FO-SQL query, we will mean a closed query satisfying the above requirements. 
Two examples are given in Figure~\ref{fig:GNSQL-examples}. We assume that the reader 
is familiar with the semantics of SQL, and hence omit the formal semantics 
of the fragment FO-SQL. We just mention that we disregard order and duplicates, 
treating relations as \emph{sets} of tuples. It is known that, under this set-based semantics, 
FO-SQL is expressively complete for first-order logic, in the sense of Codd's 
expressive completeness theorem \cite{AHV95,Libkin03}. 
That is, FO-SQL queries have the same expressive power as the domain-independent 
fragment of first-order logic. Since FO-SQL queries are defined in terms of named schemas, 
while the syntax of first-order logic is based on unnamed schemas in which the 
attributes of a relation are identified by natural numbers instead of by attribute names, 
here, we consider a FO-SQL query $q$ of type $\{A_1, \ldots, A_n\}$ to be ``equivalent'' 
to a first-order formula $\phi(x_1, \ldots, x_n)$, containing the relation names $\textsc{rel}_i$ 
occurring in $q$ as relation symbols of appropriate arity, if for every instance $I$ 
and for every $n$-tuple $\textbf{a}$, the tuple $\textbf{a}$ is an answer to $q$ in $I$ 
if and only if $I\models\phi(\textbf{a})$.

The most important features of (full) SQL that are excluded in the above definition of FO-SQL 
are \emph{constants}, \emph{arithmetical comparison}, and \emph{aggregation}. 
We will discuss the importance of these restrictions later in Section~\ref{sec:discussion}.

\vskip -.4em
\paragraph*{GN-SQL:\ the Guarded-Negation Fragment of FO-SQL}

\begin{figure} 
\begin{tabular}{@{}l@{}}
\begin{tabular}{@{}l@{}}
\textsf{select} $A.\textsc{name}$ \textsf{from} $\textsc{author}~ A$ \textsf{where not exists}( \\
~ \textsf{select} $B.\textsc{title}$ \textsf{from} $\textsc{book}~ B$ \textsf{where} $B.\textsc{auth} = A.\textsc{name}$)
\end{tabular}
\\ \\
\begin{tabular}{@{}l@{}}
\textsf{select} $A.\textsc{name}$ \textsf{from} $\textsc{author}~ A$ \textsf{where not exists}( \\
~ \textsf{select} $B.\textsc{title}$ \textsf{from} $\textsc{book}~ B$ \textsf{where not} $B.\textsc{auth} = A.\textsc{name}$)
\end{tabular}
\end{tabular}
\caption{Two examples of FO-SQL queries, the first negation-guarded and the second not.
\label{fig:GNSQL-examples}}
\end{figure}

We say that an SQL query is \emph{negation-guarded} if the following two conditions hold: 
\begin{enumerate} \setlength{\itemsep}{0em}
\item each \textsf{except} operator has as its first argument a simple projection
      and as its second argument an uncorrelated subquery.
\item each \textsf{not} operator has as its argument a condition 
      with at most one free tuple variable. 
\end{enumerate}
Here, by a \emph{simple projection}, we mean a \textsf{select-from-where} query, 
where the \textsf{where}-clause is `\textsf{true}'. 
GN-SQL is the fragment of FO-SQL consisting of all negation-guarded queries.

 To illustrate this definition, 
consider the two queries given in Figure~\ref{fig:GNSQL-examples}. The first query 
involves a single occurrence of negation, which is guarded, since the negated condition 
has only one free tuple variable, namely $A$. The second query, on the other hand, 
is \emph{not} a GN-SQL query, since the second occurrence of negation is not guarded. 
Indeed, the condition $B.\textsc{auth} = A.\textsc{name}$ has two free tuple variables. 

The next theorem states that GN-SQL captures GNFO, in the same way that FO-SQL captures 
full first-order logic, as we discussed above (the same conventions apply, concerning 
what it means for a FO-SQL query to be equivalent to a first-order formula).

\begin{theorem}[GN-SQL is Codd-complete for GNFO]\label{thm:Codd-GNSQL}
Each GN-SQL query can be translated in linear time into an equivalent domain-independent GNFO formula. 
Conversely, each domain-independent GNFO formula can be translated in exponential time into an equivalent GN-SQL query.
\end{theorem}

It can be shown that the exponential complexity of the translation from GNFO to GN-SQL is in general unavoidable for formulas of the form $(R(x_1)\lor S(x_1)) \land \cdots \land (R(x_n)\lor S(x_n))$.  On the other hand, the proof of Theorem~\ref{thm:Codd-GNSQL} shows that if the schema includes a unary relation $\textsc{adom}$ that is guaranteed to denote the active domain of the instance, then there is a polynomial translation. 

\begin{figure}\small
\begin{center}
\begin{tabular}{@{}l||l|l|l|l|l@{}}
  benchmark       
  & \rotatebox{90}{queries} 
  & \rotatebox{90}{\parbox{16mm}{\raggedright queries with negation$^1$}}
  & \rotatebox{90}{\parbox{16mm}{\raggedright queries with unguarded negation}} 
  & \rotatebox{90}{\parbox{16mm}{\raggedright queries with inequalities$^2$}} 
  & \rotatebox{90}{\parbox{16mm}{\raggedright queries with unguarded inequalities}} \\
\hline
TPC-H      &  22  &  4  &  0  &  3  &  1 \\ 
TPC-DS     &  99  &  8  &  1  &  8  &  7 \\
SkyServer  &  48  &  2  &  0  &  8  &  1
\end{tabular}
\end{center}

{\small $^1$ By negation, we mean any occurrence of \textsf{not} or \textsf{except}.} 

{\small $^2$ An inequality is any occurrence of $<>$ or \textsf{!=}. An inequality is \emph{guarded} if the corresponding negation \textsf{not(\ldots = \ldots)} is guarded.} 

\caption{Usage of negation in SQL benchmarks \label{fig:usage}}
\end{figure}

\subsection{Negation in Practice: a Benchmark Study}

In order to assess the usage of negation in SQL queries in practice, we have studied the workloads of two standard SQL benchmarks, namely TPC-H~\cite{TPCH} and TPC-DS~\cite{TPCDS}. These benchmarks were designed to evaluate and compare the performance of relational database management systems. In addition, we studied the sample queries published on the Sloan Digital Sky Survey (SDSS) SkyServer website \cite{SkyServer}, a selection of actual queries submitted by SDSS users.  For each query, we investigated whether the query uses negation, and, if so, whether the query is negation-guarded. We also studied the use of inequalities, and investigated which of these inequalities can be expressed using guarded negation. The results, given in Figure~\ref{fig:usage}, shows that most queries using negation use only guarded negation.  
We should note here that most queries contain SQL constructs, such as aggregation, that do not belong to FO-SQL. Therefore, the queries are not necessarily expressible in GN-SQL. The statistics in Figure~\ref{fig:usage} are only concerned specifically with the explicit use of \emph{negation}. We also did not investigated the use of other SQL constructs such as \emph{outer joins}, that can, in some sense, be viewed as involving an implicit form of negation.

\section{Guarded Negation in Datalog}
\label{GNDsec}

In this section, we present a powerful variant of Datalog with 
stratified negation, which we call GN-Datalog and which, in terms of
its expressive power, is contained in GNFP.
We first briefly recall the syntax and semantics of Datalog, 
with and without stratified negation.

\newcommand{\IDB}{\textup{IDB}\xspace}
\newcommand{\EDB}{\textup{EDB}\xspace}
\newcommand{\Rules}{\textup{Rules}\xspace}

\begin{definition}[Datalog]
A \emph{Datalog program} is specified by a triple $\Pi = (\EDB^{\Pi},\IDB^{\Pi},\Rules^{\Pi})$, where ${\EDB}^\Pi$ and ${\IDB}^\Pi$ are disjoint sets of relation names, each with an associated arity, and $\Rules^\Pi$ is a finite set of rules of the form 
\[ \phi \leftarrow \psi_1, \ldots, \psi_n\] where $\phi, \psi_1, \ldots, \psi_n$ are atomic formulas of the form $R(x_1, \ldots, x_n)$ with $R\in{\EDB}^\Pi\cup {\IDB}^\Pi$ and $x_1, \ldots, x_n$ a sequence of first-order variables of appropriate length. We refer to $\phi$ as the head of the rule, and $\psi_1, \ldots, \psi_n$ as the body of the rule. 
In addition, we require that (i) every first-order variable occurring in the head 
of a rule must occur in the body, and (ii) the relation in the head of each rule 
must be an IDB relation.

A Datalog query is a pair $(\Pi,Ans)$, where $\Pi$ is a Datalog program and $Ans$ is a union of conjunctive queries over the schema $\EDB^{\Pi}\cup\IDB^{\Pi}$.
The semantics of a Datalog query is defined as follows: first,
if $\Pi$ is a Datalog program, $I$  an instance over the schema $\EDB^{\Pi}$, and $k$ a natural number, then we denote by 
$\Pi^k(I)$ the instance over the schema $\EDB^{\Pi}\cup\IDB^{\Pi}$ containing all facts that can be derived from the facts in $I$ using at
most $k$ rounds of applications of rules of $\Pi$. In addition, we denote by $\Pi^{\infty}(I)$ the union $\bigcup_k \Pi^k(I)$. 
If $q=(\Pi,Ans)$ is a Datalog query and $I$ an instance over the schema $\EDB^{\Pi}$, then we denote by 
$q(I)$ the set of all tuples that are an answer to the query $Ans$ in $\Pi^{\infty}(I)$.
\end{definition}

We remark that the above definition differs slightly from the standard presentation of Datalog. Usually, $Ans$ is required to be a designated relation from $\IDB^{\Pi}$ instead of a union of conjunctive queries. The presentation we use here is convenient as it helps simplify the definitions below. On the other hand, note that this is not essential: a Datalog program can always be extended with an additional IDB relation and with additional rules computing the $Ans$ query.

\begin{definition}[Datalog with Stratified Negation]
A \emph{Datalog$^\neg$ program} is a Datalog program $\Pi$ where the body of each rule may, 
in addition, contain atomic formulas of the form $\neg R(x_1, \ldots, x_n)$ provided 
that $R\in \EDB^{\Pi}$, and provided that each first-order variable occurring in 
the head or body of the rule occurs positively in the body. 
A \emph{Datalog program with stratified negation} is a sequence $\tilde{\Pi} = (\Pi_1, \ldots, \Pi_n)$ of Datalog$^\neg$ programs, called strata, with $n\geq 1$, where for each $i=2\ldots n$, ${\EDB}^{\Pi_i} = {\EDB}^{\Pi_{i-1}}\cup{\IDB}^{\Pi_{i-1}}$. We use
${\EDB}^{\tilde{\Pi}}$ and ${\IDB}^{\tilde{\Pi}}$ to denote 
${\EDB}^{\Pi_1}$ and $\bigcup_{i=1\ldots n} {\IDB}^{\Pi_i}$, respectively. 

A \emph{Datalog query with stratified negation} is a pair $(\tilde{\Pi},Ans)$,  
where $\tilde{\Pi}= (\Pi_1, \ldots, \Pi_n)$ is a Datalog program with stratified negation 
and $Ans$ is a union of conjunctive queries over the schema $\EDB^{\tilde{\Pi}}\cup\IDB^{\tilde{\Pi}}$. 
The semantics of Datalog programs and of Datalog queries extends naturally 
to Datalog with stratified negation, by defining $\tilde{\Pi}^\infty(I)$ 
as $\Pi_n^{\infty}(\Pi_{n-1}^{\infty}(\cdots \Pi_1^{\infty}(I)\cdots))$ 
for $\tilde{\Pi}=(\Pi_1, \ldots, \Pi_n)$.  
\end{definition}

We say that a Datalog program $\Pi$ is \emph{non-recursive} 
if no IDB occurs in the body of any of its rules, and hence, in particular, 
for all instances $I$ we have that $\Pi^{\infty}(I)=\Pi^1(I)$.
We say that a Datalog program with stratified negation is non-recursive 
if it consists entirely of non-recursive strata. %

\begin{definition}[GN-Datalog]
A GN-Datalog program is a Datalog program with stratified negation $\tilde{\Pi}=(\Pi_1, \ldots, \Pi_n)$, where each rule
 \[\phi_0\leftarrow (\neg)\phi_1, \ldots, (\neg)\phi_n ~~~ \in \Rules^{\Pi_k}~~(1\leq k\leq n)\]
is \emph{negation guarded}, meaning that the following holds:
\begin{itemize}
\item[]
    For each atom $\phi_i$ that either occurs negated in the body or is the head, the body includes a positive atom $\phi_j$ containing all first-order variables occurring in $\phi_i$, and $\phi_j$ uses a relation from $\EDB^{\Pi_k}$.%
\footnote{\small To understand why this is the appropriate definition of negation guardedness, observe that a rule of the form $\phi\leftarrow \psi_1, \ldots, \psi_n$ expresses that $\neg\exists\textbf{x}(\psi_1\land\cdots\land\psi_n\land\neg\phi)$, i.e., the head of the rule plays the same role as a negated atom in the body. }

\end{itemize}
A \emph{GN-Datalog query} is a Datalog query with stratified negation, 
where each rule is negation guarded. Note that this requirement concerns only the rules; 
the answer query $Ans$ can be any union of conjunctive queries.
\end{definition}

\begin{theorem} \textup{\bf(Non-recursive GN-Datalog is Codd-complete for GNFO)} \label{thm:Codd-nr-GN-Datalog}
  Each non-recursive GN-Datalog query is equivalent to a
  domain-independent GNFO formula, and vice versa, via exponential 
translations.
\end{theorem}

  The translation from non-recursive GN-Datalog to GNFO
  given in the proof of Theorem~\ref{thm:Codd-nr-GN-Datalog} can be extended in a straightforward manner to a translation from GN-Datalog to the extension of GNFP with simultaneous fixed-point operators. Since simultaneous fixed-point operators can be eliminated (at the cost of an additional exponential blow-up, cf.~Section~\ref{sec:prel}), we obtain following:

\begin{theorem}%
\label{thm:Codd-GN-Datalog}
 Each GN-Datalog query is equivalent to a domain-independent alternation-free GNFP formula. 
\end{theorem}

The translation from non-recursive GN-Datalog to GNFO provided 
by Theorem~\ref{thm:Codd-nr-GN-Datalog} involves an exponential
blow-up, due to an elimination of subformula sharing. 
The translation from GN-Datalog to GNFP provided by 
Theorem~\ref{thm:Codd-GN-Datalog} involves another exponential blow-up, 
due to the elimination of simultaneous fixed-point operators. 
These sources of exponential complexity can be avoided 
(i) 
if we transcribe GN-Datalog queries into GNFP formulas 
over a larger schema (containing a relation symbol not only 
for each EDB of the GN-Datalog query, but also for each IDB), 
and (ii) freely use simultaneous fixed-point operators in the GNFP
formula. 
More precisely, the proof of Theorem~\ref{thm:Codd-GN-Datalog} can be 
adapted in a straightforward manner to show the following result, 
which will be useful later on
 (where, for two schemas $S\subseteq \widehat{S}$, an \emph{$\widehat{S}$-expansion}
of an instance $I$ over $\widehat{S}$ is an instance over $\widehat{S}$ that agrees with $I$ on all facts over $S$).

\begin{theorem}\label{thm:Codd-GN-Datalog-PTIME}
  For every $k$-ary GN-Datalog query $q$ over a schema $S$ one can
  compute in polynomial time a GNFP sentence $\phi_q$ and a GNFP
  formula 
$\psi_q(x_1, \ldots, x_k)$, 
both with simultaneous fixed point operators, and over a possibly larger schema $\widehat{S}$, such that
\begin{trivlist}
\item 1. each instance $I$  has a unique $\widehat{S}$-expansion $\widehat{I}$ satisfying $\phi_q$.
\item 2. for all instances $I$ and $k$-tuples $\textbf{a}$, $\textbf{a}\in q(I)$ iff $\widehat{I}\models\psi_q(\textbf{a})$.
\end{trivlist}
\end{theorem}

\section{Relationships with Existing Languages}\label{sec:existing}

\emph{Monadic Datalog} is a well-known Datalog fragment that combines
an interesting level of expressiveness with good algorithmic behavior 
thanks to a tight connection with tree automata, which 
also make monadic Datalog suitable for a number of applications, 
e.g.~\cite{GottlobKoch2004}. It also stands out a fragment 
for which the boundedness problem is decidable \cite{CosmadakisGaKaVa88}, 
Theorem~\ref{monDLdecthm} below. 
Monadic Datalog does not allow any form of negation and since all IDB 
predicates are unary, guardedness of rule heads is guaranteed, 
so that monadic datalog rules are trivially negation guarded. 
We will show that boundedness remains decidable for GN-Datalog.%

Datalog LIT is a fragment of stratified Datalog whose model checking 
has linear-time data complexity~\cite{GottlobGV02}. 
Each Datalog LIT rule must either contain in its body as `guard' 
a positive literal containing all variables occurring in the rule, 
or must solely be comprised of unary literals (including its head).  
While the `guard' of guarded rules need not be an EDB atom, 
\cite{GottlobGV02} shows that every Datalog LIT program can 
(in exponential time) be transformed into 
an equivalent one having only EDB atoms as guards. 
The latter are trivially negation guarded. 
Every Datalog LIT program is thus equivalent to a GN-Datalog program. 

GNFO subsumes a number of formalisms having currency in ontological reasoning, 
such as the linear- and guarded tuple generating dependencies (tgds) underlying 
the recently promoted Datalog$^\pm$~\cite{CGL09dl} framework and the more 
general frontier-guarded tgds \cite{Baget11ijcai} that subsume the description 
logics DL-Lite$_{R}$ (which captures RDF Schema), $\mathcal{ELI}$,  
and $\mathcal{ELH}^{dr}_{\bot}$ \cite{Baget11AI}, which is the core 
of the proposed OWL-EL profile of the OWL 2.0 ontology language.
GNFO can encode query answering and containment assertions involving 
such specifications of constraints or TBoxes. 
A tgd is a sentence %
\begin{equation} \label{eq:tgd}
  \forall \textbf{x},\textbf{y} \ \phi(\textbf{x},\textbf{y}) 
        \limp \exists \textbf{z} \ \psi(\textbf{y},\textbf{z})
\end{equation}
where both $\phi$ and $\psi$, called the body and the head, respectively, 
of the tgd rule, are conjunctions of positive atoms. 
When working under OWA one can assume, as a matter of convenience, 
that the head of every tgd is a single atom.  
A tgd is linear if $\phi$ consists of a single atom; it is guarded 
if $\phi$ contains as conjunct an atom $R(\textbf{x},\textbf{y})$, the `guard', 
in which all of the variables of the body occur together; and it is 
frontier-guarded if the body contains an atom $P(\textbf{y})$ in which 
all of the variables shared by the body and the head of the rule occur together. 
Every frontier-guarded tgd naturally translates to a GNFO sentence.  

Other query languages that can be viewed as fragments of GNFO 
include the semi-join algebra~\cite{Leinders05}, as well as
Unary Conjunctive View Logic (UCV) and Core XPath,
cf.~\cite{tCS11}.

\section{Query Containment}\label{sec:containment}

We now exploit the connection with GNFO and GNFP to show that query containment is decidable for GN-SQL and for GN-Datalog.  Recall that a query $q$ is satisfiable if there exists an instance $I$ such that the set of answers $q(I)$ is non-empty, and that a query $q_1$ is contained in a query $q_2$ if, for all instances $I$, $q_1(I)\subseteq q_2(I)$. The satisfiability problem can be viewed as (the complement of) a special case of the query containment problem, where the second query $q_2$ is any fixed unsatisfiable query.

\begin{theorem}\label{thm:containment}
Query containment is 2ExpTime-complete for both GN-SQL queries and GN-Datalog queries.
Hardness holds already for satisfiability of non-recursive GN-Datalog, and GN-SQL, 
over a fixed EDB schema.
\end{theorem}

The 2ExpTime upper bounds for GN-SQL follow directly from Theorem~\ref{thm:Codd-GNSQL} 
and Theorem~\ref{thm:GNFO}. The 2ExpTime upper bounds for GN-Datalog do not follow 
directly from Theorem~\ref{thm:Codd-GN-Datalog} and Theorem~\ref{thm:GNFO}, 
due to the exponential complexity of the translation from GN-Datalog to GNFP involved. 
However, it follows using Theorem~\ref{thm:Codd-GN-Datalog-PTIME}: 
let $q_1, q_2$ be $k$-ary GN-Datalog queries ($k\geq 0$), and 
let $\phi_1, \psi_1(x_1, \ldots, x_k)$ and $\phi_2, \psi_2(x_1, \ldots, x_k)$ be 
the GNFP-formulas with simultaneous fixed point operators 
obtained by Theorem~\ref{thm:Codd-GN-Datalog-PTIME}. 
We may assume without loss of generality that the only relation symbols that $\phi_1, \psi_1$ and $\phi_2, \psi_2$ have in common are the relation symbols that appear in $q_1$ and $q_2$. 
It follows that $q_1$ is contained in $q_2$ if and only if $\phi_1\land\psi_1(x_1, \ldots, x_k)\models \phi_2\to \psi_2(x_1, \ldots, x_k)$. This gives us the desired result, since, as we explained in Section~\ref{sec:prel}, the 2ExpTime upper bound for GNFP entailment from Theorem~\ref{thm:GNFO} extends to the case with simultaneous least-fixed point operators. 
The lower bounds are obtained by adapting the proof of an 2ExpTime-hardness result for a fragment of GNFO in \cite{tCS11}. 

Theorem~\ref{thm:containment} generalizes the known decidability result for monadic datalog 
and unions of conjunctive queries~\cite{CosmadakisGaKaVa88}.
In addition, it easily implies the decidability
of containment of Datalog queries in Unions of Conjunctive
Queries~\cite{ChaudhuriVardi97}. This can be seen as follows.
For each Datalog query $q$, let $\widehat{q}$ be the GN-Datalog query 
obtained from $q$ by guarding each rule using an additional conjunct 
that is a fresh EDB relation. Then for each UCQ $q'$ over the original
schema, we have that $q$ is contained in $q'$ if and only if $\widehat{q}$
is contained in $q'$. One direction follows directly from the fact
that $\widehat{q}$ is contained in $q$. For the other direction, note 
that every counterexample $I$ to the containment of $q$ in $q'$  gives
rise to a counterexample $I'$ to the containment of $\widehat{q}$
in $q'$. The instance $I'$ in question extends I by interpreting each 
new EDB relation as the total relation containing all tuples over the 
active domain of $I$.

As a direct consequence of the finite model property of GNFO \cite{BtCS11} 
and of Theorems~\ref{thm:Codd-GNSQL} and~\ref{thm:Codd-nr-GN-Datalog}, respectively,
we find that query containment is \emph{finitely controllable} 
for GN-SQL queries and for non-recursive GN-Datalog queries. 
By this we mean that one query is contained in an other on finite instances 
if, and only if, the containment holds on unrestricted instances (a finite model property). 

\begin{theorem} \label{thm:fincontrol}
Satisfiability and containment are finitely controllable for GN-SQL and 
for non-recursive GN-Datalog. %
\end{theorem}

\section{Query Answering}\label{sec:answering}

\subsection{(Closed-World) Query Evaluation}\label{sec:evaluation}

Since GN-SQL and non-recursive GN-Datalog admit translations into first-order logic, 
the \emph{data complexity} of query evaluation is in AC$^0$ for both query languages. 
Similarly, since GN-Datalog is contained in the fixed-point logic FO(LFP), 
the data complexity of query evaluation is in PTime.
In fact, there is a GN-Datalog query (a monadic Datalog query) for which 
query evaluation is PTime-hard in terms of data complexity~\cite{GottlobGV02}. 

In what follows we consider the \emph{combined complexity} of query evaluation.
Datalog evaluation is known to be ExpTime-complete for combined complexity 
(implicit in \cite{Vardi82}). Monadic Datalog evaluation is known to be 
NP-complete~\cite{GottlobKoch2004}.
The ``guarded fragment of Datalog'' (every rule contains an EDB atom containing 
all variables occurring in the rule) evaluation is in PTime \cite{GottlobGV02}. 
Non-recursive Datalog with stratified negation is PSPACE-complete~\cite{DEGV01}.

Recall that \PNPlogsq is the class of those problems 
that can be solved by a polynomial time deterministic algorithm 
that is allowed to ask $O(\log^2(n))$ queries to an NP-oracle. 
(It relates to better known complexity classes this way: 
NP $\subseteq$ DP $\subseteq$ \PNP$^{[\mathrm{log}]}$ $\subseteq$ \PNPlogsq 
$\subseteq$ $\cdots$ $\subseteq$ \PNP$^{[\mathrm{log}^i]}$ $\subseteq$ \PNP 
$\subseteq$ $\Sigma^p_2$ $\subseteq$ PSPACE $\subseteq$ EXPTIME.)

\begin{theorem}\label{thm:PNPlogsq}
The combined complexity of evaluating GN-SQL queries is \PNPlogsq-complete.
\end{theorem}

\begin{proof}
The upper bound follows directly from Theorem~\ref{thm:Codd-GNSQL} and Theorem~\ref{thm:GNFO}. 
For the lower bound, observe that the translation from GNFO to GN-SQL given in the proof 
of Theorem~\ref{thm:Codd-GNSQL} is polynomial in the presence of a unary relation $\textsc{adom}$ 
containing all elements in the active domain. We may assume without loss of generality 
that our input instance contains such a relation. Therefore, the lower bound from 
Theorem~\ref{thm:GNFO} extends to GN-SQL as well.
\end{proof}

We show here that the same problem is \PNP-complete for GN-Datalog. 
Recall that the best known upper bound on the complexity of model checking GNFP 
is $\NP^\NP\cap\coNP^\NP$.

\begin{theorem}\label{thm:GNDatalog-PNP}
The combined complexity of evaluating GN-Datalog queries is \PNP-complete.
Hardness holds already for non-recursive GN-Datalog queries with only 
unary IDB predicates and nullary negation. 
\end{theorem}

\subsection{Open-World Query Answering}

Open world (OWA) query answering is the following problem: 
\emph{given a query $q$, an instance $I$, and a tuple of values $\textbf{a}$, 
decide whether it is the case that $\textbf{a}$ belongs to the answers of $q$ 
in every instance extending $I$ with additional facts.}
An instance of open-world query answering $I \models_\mathrm{OWA} q(\textbf{a})$ 
thus asks for the unsatisfiability of $I \cup \{ \lnot q(\textbf{a}) \}$ in the usual 
first-order semantics, when treating $I$ as a set of atomic facts with its elements as constants.
Open world semantics is the natural choice when working with incomplete databases, 
in data exchange settings, and in the context of ontological reasoning. 
In each of these settings, open world query answering is an extensively 
researched problem. 

\newcommand{\OWA}{\textrm{OWA}}

In this section, we investigate the data complexity of open-world query answering 
for queries with guarded negation. Formally, for each query $q$ we denote by $\OWA_q$ 
the problem, given an instance $I$ and a tuple of values $\textbf{a}$ from $\adom(I)$, 
to decide whether $I\models_\OWA q(\textbf{a})$. 
More generally, for each query $q$ and for each set of constraints $\Sigma$, 
we denote by $\OWA_{q,\Sigma}$ the problem, given an instance $I$, to decide 
whether $I,\Sigma\models_\OWA q$. 

Note that, in the absence of constraints, for conjunctive queries $q$, 
by monotonicity, the problem $I\models_\OWA q(\textbf{a})$ coincides 
with $I\models q(\textbf{a})$, and therefore $\OWA_q$ is in PTime (in fact, in AC$^0$).
For first-order queries $q$, on the other hand, the problem $\OWA_q$ can be undecidable. 
We will show below that the problem is decidable for first-order queries with guarded negation.

As constraints, we will consider tuple-generating dependencies, cf.~\eqref{eq:tgd}, 
and key constraints. As noted above, linear-, guarded- and frontier-guarded 
tgds~\cite{Baget11ijcai,CGL09dl} are expressible in GNFO. 
With respect to OWA query answering, conjunctive queries are known to be FO-rewritable 
relative to linear tgds \cite{CGL09dl} and possess Datalog rewritings relative to 
frontier-guarded tgds \cite{Baget11rep}. 
Accordingly, the data complexity of open-world query answering 
for conjunctive queries against linear tgds 
is in $\mathrm{AC}_0$, in PTime for frontier-guarded tgds, and PTime-complete 
already for guarded tgds \cite{CGL09dl}. 

We begin by observing that OWA query answering for GNFO queries, 
as for many description logics~\cite{Ros06icdt,CGLLR06kr,OCE06dl}, 
has coNP data complexity. 
For an instance $I$, we denote by $|I|$ the  total number of facts of $I$, and, for 
two instances $I,J$, we write $I\subseteq J$ if every fact of $I$ is a fact of $J$. 

\begin{proposition} \label{prop:OWA_linModel}
Let $\phi(\textbf{x})$ be a fixed GNFO formula. 
For an instance $I$ and a tuple $\textbf{a}$ of elements from $\adom(I)$, 
if there is an instance $J\models\phi(\textbf{a})$ with $I\subseteq J$, 
then there is an instance $J\models\phi(\textbf{a})$ with $I\subseteq J$ and $|J|=O(|I|)$. 
\end{proposition}

Proposition~\ref{prop:OWA_linModel} tells us that, in solving the open-world query answering problem for GNFO queries, it suffices to consider ony instances whose
size is linear in the size of the input instance. This gives us the following:

\begin{theorem}
For each GNFO query $q$ (in particular, for each GN-SQL query),
$\OWA_q$ is in coNP. 
There is a boolean GN-SQL query $q$ for which $\OWA_q$ is coNP-hard.
\end{theorem}

\begin{proof} 
The coNP upper bound is immediate from the above proposition.
Given an instance $I$ with distinguished elements $\textbf{a}$,  
Proposition~\ref{prop:OWA_linModel} shows that to refute 
  $I \models_\mathrm{OWA} q(\textbf{a})$ 
it suffices to guess a linear size instance $J$ with $I\subseteq J$ 
and test in polynomial time that $J$ satisfies $\lnot q(\textbf{a})$.
The lower bound is established by a reduction from 3-colorability. 
Let $q$ be the GNFO sentence (for readability, we omit the repeated 
occurrences of $Nx$ as guard): \vspace{-1mm}
\begin{equation} \label{eq:threecol} \small
 \exists x (Nx \land \lnot P_1x  \land \neg P_2x \land\neg P_3x ) \ \lor \ 
 \bigvee_i \exists xy ( Exy \land P_ix \land P_iy ) ~~~
 \vspace{-1mm}
\end{equation} 
expressing that $P_1,P_2,P_3$ do not constitute a valid 3-coloring of the graph $(N,E)$.
It is easy to check that a simple undirected graph $G$ is not 3-colorable 
iff $G \models_\mathrm{OWA} q$, and it is straightforward to formulate 
the domain-independent boolean query~\eqref{eq:threecol} in GN-SQL. 
\end{proof}

This is remarkable, given that open-world query answering is in general undecidable for first-order  queries, even in the absence of constraints. 

Recall that every frontier-guarded tgd can be formulated as a GNFO sentence.
This allows us to lift the above result to the open-world query answering problem 
with constraints that are frontier-guarded tgds. More precisely, if $\Sigma$ is a 
set of frontier-guarded tgds, then $\OWA_{q,\Sigma}$, by definition, coincides with
$\OWA_{q\lor\bigvee_{\sigma\in\Sigma}\neg\sigma}$, and therefore we get the following.

\begin{corollary}
For each GNFO query $q$ and for each finite set of frontier-guarded tgds $\Sigma$,
$\OWA_{q,\Sigma}$ is in coNP. 
\end{corollary}

In various contexts, such as data exchange \cite{FKMP05}, it is useful to consider
incomplete databases that contain, besides constant values, also
labeled null values. In this case, open world query answering is
defined not in terms of extensions of instances, but in terms of homomorphisms that are allowed to map the labeled null values to constant values or to other labeled null values. It is worth observing that the above proofs go through in this more general setting with null values, showing that for GNFO queries $q$ and for finite sets of frontier-guarded tgds $\Sigma$,
$\OWA_{q,\Sigma}$ is in coNP even over instances containing labeled nulls.

Next we identify a subfragment of GNFO that accommodates the earlier mentioned 
formalisms including conjunctive queries and frontier-guarded tgds 
and whose queries enjoy PTime data complexity for OWA.
Recall that open-world query answering $I \models_\mathrm{OWA} q$ asks for 
the unsatisfiability of $I \cup \{\lnot q\}$. 
Under negation, the subformula  
$ \exists x (Nx \land \lnot P_1x \land\neg P_2x \land\neg P_3x) $ 
of the coNP-complete query~\eqref{eq:threecol} turns into the
disjunctive requirement $\forall x(Nx \limp P_1x\lor P_2x \lor P_3x)$ 
that is, in an intuitive sense, ultimately responsible for intractability. 
Indeed, it has been observed in the context of DL-Lite that the 
introduction of even the weakest form of disjunction renders 
query answering intractable~(see, e.g.,~\cite{CGLLR06kr}). It turns out that the positive 
occurrence of conjunctions $\lnot A(x) \land \ldots \land \lnot B(x)$ 
involving two or more negated conjuncts %
are the only source of intractability in GNFO queries.

\begin{definition}[serial GNFO queries, SGNQ] ~\\
A GNFO-formula $\varphi$ is \emph{serial} if it is in DNF and
no conjunction $\lnot \chi(x) \land \ldots \land \lnot \psi(x)$ 
with two or more negated conjuncts occurs positively %
in $\varphi$, i.e., in the scope of an even number of negations.
Let SGNQ denote the set of serial GNFO queries. 
\end{definition}

Clearly, every union of conjunctive queries is a serial GNFO query. %
Furthermore, every frontier-guarded tgds, as well as its negation, is equivalent to a boolean serial GNFO queries. 
It fact, for every finite set $\Sigma$ of frontier-guarded tgds and for every serial GNFO query $q$, we have that
$q \lor \bigvee_{\sigma\in\Sigma} \lnot \sigma$ is a serial GNFO query. In other words, the
reduction from open-world query answering in the presence of frontier-guarded tgds to open-world query answering in the absence of tgds, that we gave earlier, holds also in the case of serial GNFO queries.

\begin{theorem} \label{thrm:OWA_SGNQ_PTIME} 
For each SGNQ $q$ and for each finite set $\Sigma$ of frontier-guarded tgds, $\OWA_{q,\Sigma}$ is in PTime.  

In fact, for every boolean SGNQ $q$ we can effectively compute a boolean Datalog query $q'$ such that for all instances $I$, we have $I \models_{\mathrm{OWA}} q \, \iff I\models q'$ .

There is a boolean SGNQ query $q$ for which $\OWA_q$ is PTime-complete.
\end{theorem}

The proof is based on a reduction from the open-world 
query answering problem for SGNQs in the presence of frontier-guarded tgds to the open-world query answering problem for 
conjunctive queries in the presence of frontier-guarded tgds. A PTime solution of the latter problem 
via Datalog rewritings is due to~\cite{Baget11rep}. %

Finally, we show that OWA answering GNFO queries under key constraints is undecidable. 
This holds even for a fixed GNFO query and a fixed key constraint %
of the form 
$\forall \textbf{x} y z ( F(\textbf{x},y) \land F(\textbf{x},z) \limp y=z )$  
with $F$ a relation symbol. 

\begin{theorem} \label{thrm:OWA_SGNQineq_undec} 
\begin{enumerate} \setlength{\itemsep}{-1mm}
\item[(i)]
There is a boolean conjunctive query $q$ and a set $\Sigma$ comprising guarded tgds 
and a single key constraint, so that $\OWA_{q,\Sigma}$ is undecidable.
\item[(ii)]
There is a boolean SGNQ $q$ and a key constraint $\sigma$,
so that $\OWA_{q,\{\sigma\}}$ is undecidable.
\end{enumerate}
\end{theorem}

While undecidability of the uniform problem (where the query is part of the input) 
follows from various similar results for weaker formalisms \cite{Ros06icdt}, 
for a fixed query this seems to be a new result.

\section{Boundedness and First-Order Definability}\label{sec:boundedness}

\newcommand{\nc}{\newcommand} 
\nc{\FO}{\mathrm{FO}}
\nc{\GF}{\mathrm{GFO}}
\nc{\GSO}{\mathrm{GSO}}
\nc{\GN}{\mathrm{GN}}
\nc{\GNF}{\mathrm{GNFO}}
\nc{\bT}{\begin{theorem}} 
\nc{\eT}{\end{theorem}}
\nc{\bD}{\begin{definition}}
\nc{\eD}{\end{definition}}
\nc{\bC}{\begin{corollary}}
\nc{\eC}{\end{corollary}}
\nc{\bL}{\begin{lemma}}
\nc{\eL}{\end{lemma}}
\nc{\bP}{\begin{proposition}}
\nc{\eP}{\end{proposition}}
\nc{\bR}{\begin{rrem}}
\nc{\bO}{\begin{observation}}
\nc{\eO}{\end{observation}}
\nc{\prf}{\begin{proof}}
\nc{\eprf}{\end{proof}}
\nc{\BDD}{\mathrm{BDD}}
\nc{\CBDD}{\mathrm{BDD}_c}
\nc{\FBDD}{\mathrm{BDD}_f}
\nc{\uBDD}{\mathrm{BDD}^1}
\nc{\abar}{\mathbf{a}}
\nc{\bbar}{\mathbf{b}}
\nc{\xbar}{\mathbf{x}}
\nc{\ybar}{\mathbf{y}}
\nc{\zbar}{\mathbf{z}}
\nc{\Ibar}{\mathbf{I}}
\nc{\Xbar}{\mathbf{X}}
\nc{\Pbar}{\mathbf{P}}
\newenvironment{romanenumerate}%
{\begin{list}{(\roman{enumi})}{\usecounter{enumi}
\setlength{\labelwidth}{2cm}
\setlength{\itemindent}{0pt}
\setlength{\itemsep}{0.5\itemsep}
\setlength{\topsep}{\itemsep}
\setlength{\parsep}{0pt}
}}{\end{list}}
\nc{\bre}{\begin{romanenumerate}}
\nc{\ere}{\end{romanenumerate}}

In this section, we study the boundedness problem for GN-Datalog.  Our main result, Corollary~\ref{cor:full-boundedness}, states that it is decidable whether a GN-Datalog program is fully bounded, i.e., whether, for every instance, the computation of each stratum of the GN-Datalog program reaches a fixed point in a bounded number of steps.

The semantics of a Datalog program 
$\Pi$ can be defined in terms of a \emph{least fixed point} for the 
IDB predicates. For this we view $\Pi$, or rather each of its instatiations 
$\Pi_I$ over a given instance $I$,
as a monotone operator.
An application of this operator to any instantiation of the IDB predicates
produces the result of firing all rules once and in parallel,
on these IDB predicates and the static EDB predicates
as given in $I$. This operator  $\Pi_I$ is monotone, and the 
desired interpretation of the IDB predicates in $\Pi^\infty(I)$ 
is its unique least fixed point. This view
extends to not necessarily finite instances $I$, where $\Pi_I$, 
due to its monotonicity, still has a unique least fixed point, also refered 
to as $\Pi^\infty(I)$. As in the case of finite instances, this fixed point is 
obtained as the limit of the monotone sequence of stages 
$\Pi_I^\alpha$ generated by iterating $\Pi_I$ as an update operator, starting 
from the empty instantiation for all IDB predicates in stage $0$, and 
taking unions at limit ordinals, until finally (for cardinality reasons) 
a stage $\Pi_I^\alpha$ is reached that is a fixed point, and indeed the 
unique least fixed point ($\Pi_I^{\alpha +1} = \Pi_I^{\alpha}$ implies 
$\Pi_I^{\alpha} = \Pi^\infty(I)$).

All these considerations hold for any 
notion of program or recursion scheme that shares the crucial monotonicity 
with Datalog programs. Monotonicity refers to monotonicity 
in the IDB arguments, and is guaranteed 
by syntactic positivity in the IDB predicates in all cases we consider.%
We are mostly interested in IDB-positive $\GNF$-programs, which we 
first investigate in isolation, towards understanding their 
stratified, and overall no longer monotone, use in $\GN$-Datalog
(cf.~Definition~\ref{GNFprogramdef} below). 

The notion of \emph{boundedness} captures 
the semantic and procedural essence of non-recursive behavior 
(in contrast with syntactic non-recursiveness as defined in 
Section~\ref{GNDsec}, which focuses on a trivial reason for boundedness). 

\bD
\label{boundedprogrdef}
A monotone program $\Pi$ is \emph{c-bounded} (bounded in the classical
sense, or over unrestricted instances) if there 
exists some $n \in \mathbb{N}$ such that $\Pi_I^{n+1} = \Pi_I^n$ 
for every finite or infinite instance $I$.
It is bounded over a class of instances $\mathcal{I}$ if 
there is such an $n$ that is good for all $I \in \mathcal{I}$.
We call $\Pi$ \emph{f-bounded} if it is bounded over the 
class of \emph{all finite} instances.
\\
$\BDD(\mathcal{P},\mathcal{I})$ stands for the 
\emph{boundedness problem} for programs from $\mathcal{P}$ 
over instances from $\mathcal{I}$: given $\Pi \in \mathcal{P}$, decide whether 
$\Pi$ is bounded over $\mathcal{I}$.  
We reserve the names 
$\CBDD(\mathcal{P})$ 
and $\FBDD(\mathcal{P})$ 
for $\BDD(\mathcal{P},\mbox{\sc{All}})$ and 
$\BDD(\mathcal{P},\mbox{\sc{Fin}})$,
where $\mbox{\sc{All}}$ and 
$\mbox{\sc{Fin}}$ are the classes of
unrestricted and of finite instances, 
respectively.
\eD

Despite its basic nature, the boundedness problem is known to be undecidable 
for even very rudimentary classes of programs -- a fact which frustrated all 
hopes to systematically eliminate bounded, i.e.\ spurious, recursion in effective 
tools for query optimization.
See for instance~\cite{HillebrandEtAl95} for the undecidability of (f-)boundedness 
for Datalog programs with binary IDB predicates, as well as for Datalog programs 
with just monadic IDB predicates but with EDB negation or even just with inequalities 
in the bodies. One of the few major decidability results is the following 
from~\cite{CosmadakisGaKaVa88}.

\bT[Cosmadakis,Gaifman,Kanellakis,Vardi]
\label{monDLdecthm}
$\FBDD(\mathcal{P}) = \CBDD(\mathcal{P})$ is decidable for the 
class $\mathcal{P}$ of all monadic Datalog programs.   
\eT

The following result from classical modelö theory 
is of fundamental importance for links between boundedness 
and first-order ($\FO$) definability. It speaks about
IDB-positive programs $\Pi$ that are first-order
in the sense that the bodies of rules can be expressed in $\FO$,
by formulas that are positive in all IDB predicates 
(which guarantees monotonicity).
We use the term \emph{first-order programs} in this sense.
We say that the fixed point of $\Pi$ is $\FO$-definable 
over the class $\mathcal{I}$ if each IDB predicate 
in the least fixed point 
$\Pi^\infty(I)$ is definable in terms of the EDB predicates
by some first-order formula, uniformly across all 
$I \in \mathcal{I}$. 
  
\bT[Barwise--Moschovakis \cite{BarwiseM78}]
\label{BMthm}
An IDB-positive first-order program $\Pi$
is bounded in the classical sense
if, and only if,  
the fixed point of $\Pi$ is $\FO$-definable 
over the class of all (finite and infinite) 
instances.
\eT

Analogous equivalences can be derived for many natural 
fragments $\mathrm{L}\subseteq \FO$, where boundedness of
IDB-positive $\mathrm{L}$-programs is equated with 
$\mathrm{L}$-definability of their fixed points. 
This is true in particular also for the guarded negation 
fragment $\GNF \subseteq \FO$. 

Moreover, for many well-behaved fragments $\mathrm{L} \subseteq \FO$ 
there are model theoretic 
transfer results that say that an $\mathrm{L}$-program $\Pi$
is bounded over $\mathcal{I}$ if, and only if, it is bounded over some 
subclass $\mathcal{I}_0 \subseteq \mathcal{I}$. A case of 
particular interest is a \emph{finite model property} for boundedness,
which links the classical notion to its finite model theory version.
This, too, is available in the case of $\GNF$.  

\bD
\label{GNFprogramdef}
A $\GNF$-program is an IDB-positive program $\Pi$ with 
rules of the form 
\[
X \xbar_{s} \;\leftarrow\; \alpha_{s}(\xbar_{s}) 
\wedge \phi_{s}(\Xbar,\xbar_{s})
\]
where $\phi_{s} \in \GNF$ is positive in the IDB predicates $\Xbar$
and $\alpha_{s}$ is an EDB atom guarding the variable 
tuple $\xbar_{s}$ in the head.
\eD

The following say that for $\GNF$ we are in the ideal situation that 
f-boundedness and c-boundeness coincide, 
and that the classical and finite model theory variants of the 
Barwise--Moschovakis correspondence hold. The finite model theory analogue 
is the least straightforward of these.%
\footnote{It is known, for instance, that the universal
fragment of $\FO$, despite its finite model property, does not satisfy
this analogue: there is a purely universal program 
whose limit is uniformly definable in universal $\FO$ across all
finite instances, although it is unbounded over finite instances.}

\bP
\label{fmpobs}
For $\GNF$-programs $\Pi$ and their least fixed points $\Pi^\infty$, 
t.f.a.e.:
\bre 
\item%
$\Pi^\infty$ is $\FO$-definable 
over all finite instances.
\item%
$\Pi^\infty$ is $\FO$-definable 
over all unrestricted instances.
\item%
$\Pi^\infty$ is $\GNF$-definable 
over all finite instances.
\item%
$\Pi^\infty$ is $\GNF$-definable 
over all unrestricted instances.
\item%
$\Pi$ is bounded over all finite instances.
\item%
$\Pi$ is bounded over all unrestricted instances.
\ere%
\eP

Another crucial transfer property for $\BDD(\GNF)$ is based on the notion of \emph{treewidth}. In \cite{BtCS11}, it was suggested that the key to the good computational behavior of GNFO and GNFP lies in the fact that these logics have a \emph{tree-like model property}: for testing the satisfiability and the entailment of formulas, it suffices to consider structures of bounded treewidth. The same notion provides the key to decidability of boundedness as well.

The \emph{width} $\mathrm{w}(\Pi)$ of a $\GNF$-program $\Pi$ 
is the maximum number of element variables used in any of its 
rules in DNF.

\bL
\label{transferlem}
A $\GNF$-program $\Pi$ of width $%
\leq w$
is bounded over all unrestricted instances %
if, and only if, it is bounded over the class of 
all (possibly infinite) instances of treewidth at most $w$.
\eL

\prf
Each finite stage $\Xbar^n$ of $\Pi$ can be defined by a sequence of $\GNF$-formulas whose width is bounded by $w$. In particular, for each natural number $n$, boundedness of $\Pi$ at stage $n\geq 1$, w.r.t.~a class of structures, is equivalent to the validity of a certain $\GNF$-sentence of width $w$, on that class of structures.  Since a $\GNF$-sentence of width $w$ is valid on arbitrary structures if and only if it is valid on structures of treewidth at most $w$ (cf.~\cite{BtCS11,O2012}), the claim follows.
\eprf

We turn to decidability of $\CBDD(\GNF)$ and of 
\emph{full boundedness} (to be defined below) of $\GN$-Datalog.
Given the meager history of decidability results concerning 
boundedness for database purposes, it is interesting 
that here is one considerable extension of \emph{the} 
early decidability result for monadic Datalog from~\cite{CosmadakisGaKaVa88}, 
cf.~Theorem~\ref{monDLdecthm} above.

We note that $\GN$-Datalog is stricly more expressive than 
monadic Datalog, but avoids the dangers of negation that render 
boundedness undecidable, for instance, in the extension of 
monadic Datalog by 
just inequalities, or by negative as well as positive access 
to some binary EDB predicates.

Technically, the following decidability assertion is an easy corollary 
to the decidability results for monadic second-order logic and guarded 
second-order logic over tree-like structures in~\cite{BOW}. 
These results in turn are based on a non-trivial reduction to an automata 
theoretic decidability result of Colcombet and L\"oding, which, in the 
relevant strength needed here, has not been published yet. 
As in~\cite{BOW} we indicate this caveat formally as an assumption
(ILT), which refers to the decidability of \emph{limitedness} 
for weighted parity automata on infinite trees, as announced in 
connection with progress on earlier work in ~\cite{ColcombetLoeding08}. 

Recall that $\CBDD(\GNF)$ and $\FBDD(\GNF)$ coincide.

\bT[assuming ILT]
\label{bddGNFthm}
Boundedness for $\GNF$-programs is decidable.
\eT

\prf
The $\GNF$-formulas in an $\GNF$-program 
can be translated into explicitly guarded formulas of 
guarded second-order logic $\GSO$ (denoted $\GSO^\ast$ in~\cite{BOW}).   
The result then follows from the decidability 
of $\BDD(\GSO^\ast,\mathcal{W}_k)$, boundedness for 
$\GSO^\ast$ over the class of structures of treewidth $k$, where 
both the $\GSO^\ast$-formulas and the parameter $k$ are treated as
input (Theorem~8.8 in~\cite{BOW}). We apply 
this to the $\GSO^\ast$-translation of 
the input $\GNF$-program $\Pi$ over the class $\mathcal{W}_k$
for $k := \mathrm{w}(\Pi)$.%
\footnote{NB: since $\Pi$ really corresponds to a \emph{system} of 
least fixed points in several IDB predicates $X$, we 
need the result for systems of simultaneous fixed points 
in $\GSO^\ast$ from~\cite{BOW}, as discussed in the 
proof sketch for Theorem~11.5 there.} 
By Lemma~\ref{transferlem}, (ii), 
this is a valid reduction.
\eprf

Towards our interest in $\GN$-Datalog, with its stratified 
use of guarded negation as defined in Section~\ref{GNDsec}, 
we %
 extend the notion 
of boundedness from Definition~\ref{boundedprogrdef}
as follows. 

\bD
\label{boundedstratprogrdef}
A GN-Datalog program $\tilde{\Pi} = (\Pi_i)_{i \leq n}$ 
is called \emph{fully f/c-bounded} over  
a class of instances $\mathcal{I}$ 
if each stratum $\Pi_i$ is f/c-bounded over the class of all 
instances obtained from instances in $\mathcal{I}$ by 
evaluating all IDB predicates from lower strata 
according to $\tilde{\Pi}_{< i}$ and treating
them as EDB for $\Pi_i$. 
Equivalently, a GN-Datalog program $\tilde{\Pi} = (\Pi_i)_{i \leq t}$ is fully f/c-bounded if 
there are natural numbers $k_1, \ldots, k_t$ such that for all finite/unrestricted instances $I$,
$\Pi^\infty(I)=\Pi_t^{k_t}(\Pi_{t-1}^{k_{t-1}}(\cdots \Pi_1^{k_1}(I)\cdots))$. 
\eD

\bC[assuming ILT] \label{cor:full-boundedness}
For $\GN$-Datalog, full f-boundedness is decidable and coincides with full c-boundedness.
\eC

\prf
The proof is by induction on the number of strata. 
Note that by definition of full boundedness, 
a stratified $\GN$-Datalog program $\tilde{\Pi} = (\Pi_i)_{i \leq n}$
fails to be fully bounded if, and only if, there is a least
stratum $m \leq n$ such that $\Pi_m$ is unbounded over the class of
instances obtained by evaluating all IDB predicates of lower strata
according to $\tilde{\Pi}_{<m}$. Since these lower strata are bounded, 
this partial evaluation is in fact $\GNF$-definable.
It follows that the above arguments concerning the $\GNF$-variant 
of the Barwise--Moschovakis theorem and its finite model theory
version carry through -- stratum by stratum, and up to the 
first stratum that turns out to be unbounded, if any.
This also reduces the decidability claim to that in Theorem~\ref{bddGNFthm}.
\eprf

We remark that the passage through boundedness for $\GSO^\ast$
over $\mathcal{W}_k$, which is known to be of non-elementary complexity
even for $k=1$, 
prevents us from extracting any reasonable complexity bounds.
It is conceivable, of course, that alternative methods
yield such bounds (as is the case for other special 
cases of interest, besides that of monadic Datalog, that also 
follow from the master result of~\cite{BOW}).

\section{Discussion}\label{sec:discussion}

\subsection{Further extensions of GN-SQL}

\vskip -.4em
\paragraph*{Inequalities}
GN-SQL can be viewed as a well-behaved query language extending unions of conjunctive queries with a restricted form of negation.  In this sense, it is natural to compare GN-SQL to UCQ($\neq$), the language of unions of conjunctive queries with inequalities.  Like GN-SQL, UCQ($\neq$) is computationally well-behaved: query containment is $\Pi^p_2$-complete \cite{Klug88,Meyden97}, the combined complexity of query evaluation is NP-complete, and the data complexity of open world query answering is NP-complete w.r.t.~a large class of constraints~\cite{FKMP05}, cf.~also~\cite{Madry05}.
In this light, and in the light of Figure~\ref{fig:usage}, the question arises whether we can extend GN-SQL to allow for the use of (unguarded) inequalities. 

Let us denote by GN-SQL($\neq$) the extension of GN-SQL where conditions may make use of the inequality relation ($\neq$), but the inequality relation cannot be used to guard negations. It is easy to see that Theorem~\ref{thm:PNPlogsq} extends to GN-SQL($\neq$) --- we may view the inequality as just another relation that is part of the input instance. 
 All the other results we obtained for GN-SQL, however, fail for GN-SQL($\neq$). This follows from the fact that it is possible to express functional dependencies in GN-SQL($\neq$).
 Indeed, 
  every functional dependency $$\forall\textbf{x}\textbf{y}\textbf{z}uv(F(\textbf{x},\textbf{y},u)\land F(\textbf{x},\textbf{z},v)\to u=v)$$
  is equivalent to the GNFO sentence with inequality $$\neg\exists\textbf{x}\textbf{y}\textbf{z},u,v(F(\textbf{x},\textbf{y},u)\land F(\textbf{x},\textbf{z},v)\land u\neq v)~,$$
which can easily be expressed in GN-SQL($\neq$) as well. 
Recall that inclusion dependencies too can be expressed in GN-SQL. 
This, together with classical results in dependency theory (cf.~\cite{AHV95}) 
and Theorem~\ref{thrm:OWA_SGNQineq_undec}(ii)), implies the following:

\vskip -.4em
\begin{theorem}\label{thm:inequality-undecidable} 
\begin{enumerate} \setlength{\itemsep}{-1mm}
\item[(i)] GN-SQL($\neq$) is not finitely controllable for satisfiability or query containment. 
\item[(ii)] The satisfiability and query containment problems for GN-SQL($\neq$) are undecidable 
            (both on finite instances and on unrestricted instances).
\item[(iii)] There is a GN-SQL($\neq$) query for which open world query answering is undecidable.
\end{enumerate}
\end{theorem}

Known results for various description logics contained in GNFO 
imply that OWA answering for GNFO(${\neq}$) queries is undecidable
when the query is part of the input \cite{Ros06icdt}. Theorem~\ref{thm:inequality-undecidable}
strengthens this by showing that the problem is undecidable already 
for a fixed GNFO(${\neq}$) query. 
Naturally, similar results can be obtained for the extension of GN-Datalog with inequalities.

\vskip -.4em
\paragraph*{Constants and Comparisons}
GN-SQL queries, as we defined them, cannot contain constant values, nor arithmetical comparisons (i.e., conditions of the form $t_1<t_2$). Indeed, over linearly ordered domains, inequalities can be expressed using arithmetical comparisons ($x\neq y$ is equivalent to $x<y\lor y<x$), and hence, by Theorem~\ref{thm:inequality-undecidable}, most problems immediately become undecidable when arithmetical comparisons are allowed.
However, as we will show, our results \emph{do} generalize to the extension of GN-SQL where (i) queries may contain constant values, and (ii) arithmetical comparisons of the form $t_1<t_2$ are allowed provided that at least one of $t_1, t_2$ is a constant value.

In what follows, let $\textsc{lin}=(D,\prec)$ be any ordered domain (where  $D$ is a countable set and $\prec$ is a total order on $D$) that is ``reasonable'' in the sense that the following problems are all solvable in polynomial time (for some appropriate representation of the elements of $D$):
{%
\begin{enumerate} \setlength{\itemsep}{0em} \setlength{\topsep}{0em} \setlength{\parsep}{0em}
\item 
given 
$d_1,d_2\in D$, is it the case that  $d_1\prec d_2$?
\item given 
$d\in D$, does there exist $d'\in D$ with $d'\prec d$?
\item given 
$d\in D$, does there exist $d'\in D$ with $d\prec d'$?
\item given 
$d_1,d_2\in D$, does there exist $d'\in D$ with $d_1\prec d' \prec d_2$?
\end{enumerate}
}
Essentially, all the usual  ordered domains, such as the natural numbers $(\mathbb{N},<)$, the rational numbers $(\mathbb{Q},<)$, and the strings $(A^*,<_{lex})$ over a finite ordered alphabet $A$, are reasonable in this sense.

Let GN-SQL(\textsc{lin}) be the extension of the GN-SQL syntax where (i) all terms $t$ are allowed to be either of the form $R.\textsc{attr}$ (as before) or to be an element of $\textsc{lin}$ (in which case we call $t$ a \emph{constant}); and (ii) for all terms $t_1, t_2$ of which at least one is a constant value, $t_1<t_2$ is allowed as an atomic condition. 
The semantics of GN-SQL(\textsc{lin}) queries is only well-defined for instances whose active domain is a subset of $\textsc{lin}$. Therefore, we restrict attention to such instances.

All results for GN-SQL that we have presented can be extended to GN-SQL. For simplicity, we sketch the relevant construction here only for the query containment problem.
\begin{theorem} \label{thm:ordering-containment}
  Let $\textsc{lin}$ be any reasonable ordered domain. 
  GN-SQL(\textsc{lin}) query containment is 2ExpTime-complete.
\end{theorem}
\vskip -.6em
\paragraph*{Aggregation}
Recall that GN-SQL does not allow for any form of aggregation that is available in SQL. This is for good reason:
allowing even simple forms of aggregation such as 
counting would  quickly lead to undecidability, since query containment for unions of conjunctive queries 
under the bag semantics is undecidable~\cite{Ioannidis95:bag}.

\subsection{Further Extensions of GN-Datalog}

\vskip -.4em
\paragraph*{Allowing IDBs As Guards}
If in the definition of negation-guarded Datalog rules one permits also 
the use of IDB atoms \emph{from the same or lower stratum} as guards, 
this can result in an exponential gain in succinctness but does not increase 
the expressive power. (A simple induction on strata and on stages of 
inductive definitions of IDB predicates confirms that all tuples added 
to the interpretation of IDB predicates are guarded by some EDB atom.) 
Query evaluation complexity, however, suffers an exponential blow-up 
as a consequence of this relaxation.

\vskip -.8em
\begin{proposition}[GN-Datalog with IDB guards]\label{prop:IDB-guards}
Answering GN-Datalog queries with IDB atoms allowed as guards 
is ExpTime-complete in combined complexity. 
\end{proposition}

\vskip -.8em
\paragraph*{Capturing the Alternation-Free Fragment of GNFP} 
In \cite{GottlobGV02}, an extension of Datalog-LIT was presented, called Datalog-LITE, which includes ``generalized literals'' and was shown to capture the alternation-free fragment of guarded fixed point logic (GFP). We expect that GN-Datalog can be similarly extended, in order to subsume Datalog-LITE and capture the alternation-free fragment of GNFP.

\enlargethispage*{1mm}
\bibliographystyle{abbrv}
\bibliography{gnqueries}


\appendix

\section{Missing proofs}

\subsection{Proof of Theorem~\ref{thm:GNRA} and inexpressibility claims}

\begin{proof}
    From GN-RA expressions to GNFO formulas, there is a straightforward inductive linear translation.
    More precisely, the following table describes how to translate each GN-RA 
expression $E$ of arity $k$ into an equivalent (therefore domain-independent) 
GNFO formula $\varphi_E(x_1,\ldots,x_k)$.   

{\small\begin{tabular}{r l}
  $R$ & $R(x_1,\ldots,x_k)$ \\
  $\sigma_{i=j}(E)$ & $\varphi_E(\mathbf{x}) \land x_i=x_j$ \\
  $\pi_{i_1\ldots i_m}(E)$ & $\exists \mathbf{z} \varphi_E(\mathbf{z}) \land \bigwedge_{j=1}^m x_j=z_{i_j}$ \\
  $E \times E'$ & $\varphi_E(x_1,\ldots,x_{k_1}) \land \varphi_E(x_{k_1+1},\ldots,x_{k_1+k_2})$ \\ 
  $E \cap E'$ & $\varphi_E(x_1,\ldots,x_k) \land \varphi_E(x_1,\ldots,x_k)$ \\ 
  $E \cup E'$ & $\varphi_E(x_1,\ldots,x_k) \lor \varphi_E(x_1,\ldots,x_k)$ \\ 
  $\pi_{i_1\ldots i_m}(R) \setminus E$ & 
      $\exists \mathbf{z} \big( R(\mathbf{z}) \land \lnot \varphi_E(\mathbf{x}) \land \bigwedge_{j=1}^m x_j=z_{i_j} \big)$ \\  
\end{tabular}}

\medskip

For  the converse direction, 
    we proceed as follows: we first construct a GN-RA expression \textsc{adom} that defines the active domain of the instance 
     (the union of all unary projections of atomic relations). Next,
    for each sequence of variables $\textbf{x}=x_1, \ldots, x_k$ and for each atomic GNFO formula $\phi$ whose free variables are 
    included in $\textbf{x}$, we compute a $k$-ary GN-RA expression $\tr_{\textbf{x}}(\phi)$ that is equivalent to it under the 
    active domain semantics (i.e., over structures whose domain coincides with the active domain). 
    For instance $\tr_{x_1,x_2,x_3}(R(x_2,x_2)) = \pi_{1,2,4}\sigma_{2=3} (\textsc{adom}\times R\times \textsc{adom})$, and
    $\tr_{x_1,x_2,x_3}(x_1=x_2) = \pi_{1,1,2}(\textsc{adom}\times\textsc{adom})$.  Finally, the
    translation $\tr_{\textbf{x}}(\cdot)$ is extended to complex GNFO formulas. 
    Conjunction, disjunction and existential quantification are translated as intersection, union, and projection. Hence, the only
    remaining case is for $\tr_{\textbf{x}}(\alpha(\textbf{y})\land\neg\phi(\textbf{y}))$, where the variables in $\textbf{y}$ are included in the variables in $\textbf{x}$.
    If the guard $\alpha$ is a relational atom, $\tr_{\textbf{x}}(\alpha(\textbf{y})\land\neg\phi(\textbf{y}))$ can be defined as
the
   GN-RA expression obtained from $\tr_{\textbf{y}}(\alpha) - \tr_{\textbf{y}}(\phi)$ by (i) pulling out selections and projections as 
    necessary in order to turn the first argument of the complementation operator into a projection of an atomic relation; and (ii)
   taking a product with $\textsc{adom}$ for each variable from $\textbf{x}$ that is not included in $\textbf{y}$.

    If the guard $\alpha$ is of the form $y_1=y_2$, 
then $\tr_{\textbf{y}}(\alpha(\textbf{y})\land\neg\phi(\textbf{y}))$
is defined, in the first instance,
    as $\pi_{1,1}(\textsc{adom} - \pi_1\sigma_{1=2}\tr_{\textbf{y}}(\phi))$. Since \textsc{adom} is in general a union of all unary projections of 
  atomic relations, we need to pull out the union from the scope of the complementation operator. This is where an
    exponential blow-up may be incurred.
\end{proof}

\begin{proposition}
     The following RA expressions are not equivalent to GN-RA expressions:
\begin{enumerate}
\item $(\pi_1(R)\times S) - \pi_{1,1}(R)$
\item $\pi_{1,4}(\sigma_{2=3}(R\times R)) - R$ 
\item $\pi_{1}(R) - \pi_{1}((\pi_{1}(R)\times S)-R)$
\end{enumerate}
\end{proposition}

\begin{proof}
   In \cite{BtCS11}, the notion of \emph{GN-bisimulation} was introduced, and it was shown that GN-bisimulations preserve the truth of GNFO sentences.
   Together with Theorem~\ref{thm:GNRA} this allows us to show non-expressibility of the above RA expressions in GN-RA.  Note that 
   if any of the above RA expressions was definable in GN-RA, then also its boolean projection would be definable in GN-RA. It can be shown that
\begin{enumerate}
\item the instance $\{R(a,b),S(a),S(c)\}$, which satisfies the boolean projection  of $(\pi_1(R)\times S) - \pi_{1,1}(R)$, 
        is GN-bisimilar to the instance $\{R(a,b),S(a)\}$, which does not.
\item the instance $\{R(a,b),R(b,c),R(a,c),R(b,d)\}$, which satisfies the boolean projection  of $\pi_{1,4}(\sigma_{2=3}(R\times R)) - R$ 
        is GN-bisimilar to the instance $\{R(a,b),R(a,c),R(b,c)\}$, which does not.
\item the instance $\{R(a,b),S(a),S(c)\}$, which satisfies the boolean projection  of $(\pi_1(R)\times S) - \pi_{1,1}(R)$, 
        is GN-bisimilar to the instance $\{R(a,b),S(a)\}$, which does not. \qed
\end{enumerate}
\end{proof}

\subsection{Proof of Theorem~\ref{thm:Codd-GNSQL}}

\begin{proof}[sketch]
  Let $q$ be any (closed) GN-SQL query. We may assume without loss of generality that each tuple variable $R$ occurring in $q$ is declared in exactly one \textsf{from}-clause, and therefore has a unique associated relation name, that we will denote by $\textsc{rel}_R$. By a simultaneous induction, we can 
\begin{itemize}
\item translate each (open or closed) GN-SQL query $q$ to a GNFO formula $\phi_q(\textbf{x})$, where $\textbf{x}$ is a sequence of first-order variables, one for each
attribute name belonging to the type of the query $q$ and one for each term $R.\textsc{attr}$ where $R$ is a tuple variable that occurs freely in $q$ and 
$\textsc{attr}$ is an attribute name belonging to the type of $\textsc{rel}_R$,
\item translate each GN-SQL condition $c$ to a GNFO formula $\phi_c(\textbf{x})$, where $\textbf{x}$ is a sequence of first-order variables, one for each
term $R.\textsc{attr}$ where $R$ is a tuple variable that occurs freely in $c$ and 
$\textsc{attr}$ is an attribute name belonging to the type of $\textsc{rel}_R$.
\end{itemize}
We omit the detailed definition of the translation, which is straightforward. The clause for \textsf{not} is as follows:
    \[\phi_{\textsf{not}(\emph{condition})}(\textbf{x}) = \textsc{rel}_R(\textbf{x})\land\neg\phi_{\emph{condition}}(\textbf{x})\]
It is not hard to see that each closed GN-SQL query $q$ is equivalent to its GNFO translation $\phi_q$. In particular, this implies that $\phi_q$ is domain independent.

For the converse translation, from domain-independent GNFO formulas to GN-SQL queries, it is convenient to first assume that we have at our disposal a relation \textsc{adom} with a single attribute $A$ containing all elements belonging to the active domain. As we will show, using such a relation, it is quite straightforward to give an inductive polynomial translation from GNFO to GN-SQL. On the 
other hand, all usage of \textsc{adom} can be eliminated at the cost of an exponential blow-up. 
To see this, recall that relation names can only appear in FO-SQL queries in the \textsf{from}-clause of a \textsf{select-from-where} expression. Thus, any occurrence of $\textsc{adom}$ must be of the form
\[ \textsf{select $\alpha$ from (\ldots, \textsc{adom} $R$, \ldots) where $\beta$}\]
where, in addition, the expressions $\alpha$ and $\beta$ may refer to $R.A$.
We may equivalently replace such an expression by the union of all expressions of the following form (for all relation names $\textsc{rel}_i$ and attribute names $\textsc{attr}_j$):
\[ \textsf{select $\alpha'$ from (\ldots, $\textsc{rel}_i$ $R$, \ldots) where $\beta'$}\] 
where $\alpha'$ and $\beta'$ are obtained from $\alpha$ and $\beta$ by replacing all occurrences of $R.A$ by $R.\textsc{attr}_j$. Clearly, applying this transformation for all occurrences of \textsc{adom} yields an equivalent query that does not make use of \textsc{adom} and that is at most singly exponentially larger than the original query. 

Next, we explain how to translate domain-independent GNFO formulas to
GN-SQL queries with the help of the \textsc{adom} relation.  Let
$\phi(\textbf{x})$ be any domain-independent GNFO formula. We assume
w.l.o.g.\ that $\phi$ does not reuse any variables, and associate to each 
first-order variable $z$ a corresponding distinct tuple variable $R_z$ (whose type, in the expressions below, will consist of a single attribute named $A$). Next, we inductively translate each GNFO formula $\phi$ to a GN-SQL condition $\phi^*$, as follows.
\[\small\begin{array}{@{}l@{~}l@{~}l@{}}
  (x=y)^* &=& (R_x.A=R_y.A) \\
  \textsc{rel}(x_1, \ldots, x_n) &=& \textsf{exists(select $R.A_1$ from \textsc{rel} $R$ where} \\ 
    & &  \textsf{$R.A_1=R_{x_1}.A$ and $\ldots$ $R.A_n=R_{x_n}.A$)} \\
  (\phi\land\psi)^* &=& \phi^*\land\psi^* \\ 
  (\phi\lor\psi)^* &=& \phi^*\lor\psi^* \\ 
  (\exists x~\phi)^* &=& \textsf{exists(select $R_x.A$ from \textsc{adom} $R_x$ where $\phi^*$)} \\
  (\textsc{rel}(\textbf{x})\land\neg\phi)^* &=& \textsf{exists(select $R.A_1$ from \textsc{rel} $R$ where} \\ 
    & & \hspace{-15mm}
       \textsf{$R.{A_1}=R_{x_1}.A$ and $\ldots$ and $R.{A_n}=R_{x_n}.A$ and not($\widehat{\phi^*}$))} \\
  (x=y\land\neg\phi)^* &=& \textsf{($R_x.A=R_y.A$) and not $\phi[x/y]^*$} 
\end{array}\]
where, in the second clause and in the 6th clause, the schema of the relation \textsc{rel} 
is $\{A_1, \ldots, A_n\}$, and where, in the 6th clause,  $\widehat{\phi^*}$ is obtained from $\phi^*$ 
by replacing each term $R_{x_i}.A$ by $R.A_i$. In the last clause, $\phi[x/y]$ is the formula 
obtained from $\phi$ by replacing each free occurrence of $x$ by $y$, so that the formula 
in question has only one free first-order variable. 

Finally, starting with a GNFO formula $\phi(x_1, \ldots, x_n)$ we define the query $q_\phi$ as follows:
\[\begin{array}{l}\textsf{select $R_1.A$ as $\textsc{attr}_1$, $\ldots$, $R_n.A$ as $\textsc{attr}_n$} \\ \textsf{from \textsc{adom} $R_1$, $\ldots$, \textsc{adom} $R_n$ where $\phi^*$}\end{array}\]
where $\textsc{attr}_1$, $\ldots$, $\textsc{attr}_n$ are distinct attribute names.
It is easy to show that each domain-independent GNFO formula is equivalent to the GN-SQL query obtained from it in the above way.
\end{proof}

\subsection{Proof of Theorem~\ref{thm:Codd-nr-GN-Datalog}}
\begin{proof}
Consider any non-recursive GN-Datalog query $q=(\tilde{\Pi},Ans)$ with $\tilde{\Pi}=(\Pi_1, \ldots, \Pi_n)$. A straightforward induction on $k$ shows
  that, for every $k\leq n$ and for every $X\in \IDB^{\Pi_k}$, there is a GNFO formula $\phi$ that defines the relation 
  computed by $X$. The GNFO formula in question can be obtained by taking the disjunction of all bodies of rules that have
  $X$ in the head, replacing all occurrences of IDBs $Y\in \IDB^{\Pi_\ell}$ with $\ell<k$ by their (previously obtained)  defining GNFO formulas. 
  Finally, by taking the query $Ans$ and replacing each IDB $X\in\IDB^{\Pi_n}$ by its defining GNFO formula, we obtain a
  GNFO formula that is equivalent to $q$, and, in particular, domain independent, since $q$ is domain independent.

  Conversely, let $\phi(\textbf{x})$ be a domain-independent GNFO
  formula. We may assume that $\phi(\textbf{x})$ is in DNF. This may
  require an exponential blow-up.  First, we construct a GN-Datalog
  program with a unary IDB $\textsc{adom}$ that computes the active
  domain, as well as a binary IDB $X_{=}$ that computes 
the relation $\{(x,x)\mid x\in\textsc{adom}\}$, which will be used for translating equality statements. We omit the construction, which is straightforward. Next, by induction, we construct for each subformula  of $\phi$ that is of the form $\alpha\land\neg\chi$ a non-recursive GN-Datalog program with IDB relations $X_{\alpha\land\chi}$ and $X_{\alpha\land\neg\chi}$ computing the relation defined by $\alpha\land\chi$ and $\alpha\land\neg\chi$
  under the active-domain semantics (i.e., on structures whose active domain is the entire domain). In particular, if $\alpha$ is a relational atom, then the program includes the rule $$X_{\alpha\land\neg\chi}(\textbf{x})\leftarrow \alpha(\textbf{x})\land\neg X_{\alpha\land\chi}(\textbf{x})$$
If $\alpha$ is an equality statement, we proceed similarly, using the $X_=$ IDB  relation introduced above as a guard.

Finally, if $\phi(\textbf{x})$ is of the form 
  $\phi_1(\textbf{x}_1)\lor\ldots\lor\phi_n(\textbf{x}_n)$, 
we define $Ans$ to be the union of the conjunctive queries 
  $Ans_i(\textbf{x}) := (\, X_{\phi_i}(\textbf{x}_i) \land 
                       \bigwedge_{x\in\textbf{x}}\textsc{adom}(x) \,)$.
From the domain-independence of $\phi(\textbf{x})$, we obtain that $q=(\tilde{\Pi},Ans)$ is equivalent to $\phi$.
\end{proof}

\subsection{Proof of Theorem~\ref{thm:containment}}

It was shown in \cite{tCS11} that satisfiability for UNFO is 2ExpTime-hard, where UNFO is a syntactic fragment of GNFO.  As we explain below, the construction can be adapted to prove the same lower bound result already for GNFO formulas in DNF.  In addition, we can easily ensure that the GNFO formulas in question are domain independent, and force the existence of a fixed unary predicate denoting the active domain.  Under these restrictions, the translation from GNFO to GN-SQL (Theorem~\ref{thm:Codd-GNSQL}) and the translation from GNFO to non-recursive GN-Datalog (heorem~\ref{thm:Codd-nr-GN-Datalog}) both runs in polynomial time. Hence, we obtain 2ExpTime-hardness for satisfiability, and therefore also for query containment, for GN-SQL and non-recursive GN-Datalog.

\begin{proposition}\label{prop-sat-lower-bound}
  There is a fixed schema such that the satisfiability problem for GNFO formulas in DNF is
  2ExpTime-hard, both on arbitrary structures and on finite structures.
\end{proposition}
\begin{proof}[sketch]
  The same result, without the DNF requirement, was shown in \cite{tCS11}, in the context of a fragment of GNFO called UNFO.
  We briefly sketch the construction used in the proof in \cite{tCS11}, and
  explain how it can be adapted to use only GNFO formulas in DNF.

  Fix an alternating $2^n$-space bounded Turing machine $M$ whose word problem
  is 2ExpTime-hard. Let $w$ be a word in the input alphabet of $M$.  We
  construct a formula $\phi_{w}$ that is satisfiable if and only if $M$ accepts
  $w$. Moreover, if $\phi_{w}$ is satisfiable, then in fact it is satisfied in
  some finite tree structure.  In this way, we show that the lower bound holds
  not only for arbitrary structures, but also for finite trees and for any
  class in-between. The formula $\phi_{w}$ describes an (alternating) run of
  $M$ starting in the initial state with $w$ on the tape, and ending in a final
  configuration.

  The run is encoded as a big tree whose nodes correspond to configurations,
  whose child relation correspond to successive configurations, and where each
  node has in addition a small subtree of height $n$ attached to it, that is
  used to describe the tape content at that configuration. Here is an
  illustration of a configuration with two successor configurations (but we
  allow more than two successor configurations):
\begin{center}
\includegraphics{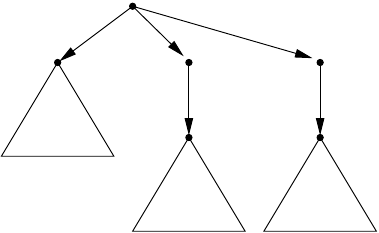}
\end{center}

Each small subtree has depth exactly $n$. The internal nodes of the subtree are
label with a unary predicate $P$. Hence a path from its root to one of its leaf
correspond to a bit string of length $n$ denoting a position of the tape. The
label of the leaf codes the content of the tape at that position.

The formula first enforces that the small subtrees have the desired structure
and that all positions are realized in at least one leaf of each small subtree. Since we
don't have inequality, we cannot force that it is realized exactly once, but we
can force in GNFO that all nodes where it is realized satisfy the same relevant unary predicates $A$:

\[\begin{split}\neg\exists xy(\text{leaf}(x)\land \text{leaf}(y)\land
x\uparrow^n\downarrow^n y \land~~~~~ \\~~~~~ \bigwedge_i (P_i(x) \leftrightarrow P_i(y)) \land
A(x) \land\neg A(y))
\end{split}\]

Here, $x\uparrow^n\downarrow^n y$ is a short for the GNFO formula describing
the fact that there is a path of the form $\uparrow^n\downarrow^n$ from $x$ to
$y$, $\text{leaf}(x)$ is a short for $\neg\exists y Rxy$, and $P_i(x)$ is a
shortcut for $\exists y(x\uparrow^{n-i}y\land P(y))$.

The following formula $\text{suc}(x,y)$ expresses that $x$ and $y$ denote
the same tape position in successive configurations, and it uses
only unary negation:
\[\text{suc}(x,y) ~:=~ \text{leaf}(x)\land \text{leaf}(y)\land (x\uparrow^{n+1}\downarrow^{n+2}y)\land
\bigwedge_i(P_i(x)\leftrightarrow P_i(y))\]
Note that the first half of the formula says that $x$ and $y$
are tape cells of successive configurations.  
Using this formula, we can specify all relevant properties of the run
(the encoding will involve formulas of the form $\forall x(\text{leaf}(x)\land \phi(x) \to \exists y(\text{suc}(x,y) \land \psi(y)))$). 

This concludes the outline of the construction used in \cite{tCS11} for showing 2ExpTime-hardness of the satisfiability problem for (a fragment of) GNFO.

The above proof clearly uses GNFO formulas that are not in DNF, and the straightforward way to bring the formulas in DNF, by ``pulling out disjunction'', would lead to formulas whose length is exponential in $n$. The problem, here, lies in the formula
$\bigwedge_i (P_i(x) \leftrightarrow P_i(y))$ expressing that two leaf nodes, in the same configuration or in successor configurations, encode the same memory location. 
This use of disjunction can be avoided using a construction from \cite{Bjorklund08}. In particular, 
we enrich our encoding of Turing machine configurations as follows: to each node $x$ of the structure, we attach a small substructure consisting of nodes that we mark with a fresh unary predicate $Q$ in order to distinguish them from the nodes that belong to the ``main structure'' (i.e., the structure as it was before adding all these new small substructures). 
The exact substructure that we attach to a node $x$ depends on whether or not the node satisfies $P$. If a node $x$ it satisfies $P$, we create a new node $y$ and add edges $E(x,y)$ and $R(x,y)$. If, on the other hand, $x$ does not satisfy $P$, we create new nodes $y$ and $z$ and add edges $E(x,z), R(x,y), R(y,z)$. Here, $E$ is a new binary predicate. 
This modification of the structure has the consequence that we can avoid the use of disjunction in comparing whether two leaf node encode the same memory location: 
suppose that $x$ and $y$ be leaf nodes of the same configuration subtree. Then 
$(P_i(x) \leftrightarrow P_i(y))$ can be equivalently expressed as 
$$\exists uu'vv'(x\uparrow^{n-i}u\land E(u,u')\land y\uparrow^{n-1}v\land E(v,v')\land u'\uparrow^{i+2}\downarrow^{i+2} v')$$
In a similar way, we can express, without using disjunction, the fact that two nodes encode the same memory location in \emph{successive} configuration subtrees. We omit the details.
\end{proof}

\subsection{Proof of Theorem~\ref{thm:GNDatalog-PNP}}

\begin{proof} 
  Upper bound: for each IDB except possibly the answer IDB, the number of tuples 
that may end up in the extension of the IDB is bounded by the number of tuples 
belonging to the extension of the EDBs, times the number of rules of the Datalog program, because each rule is guarded by an EDB (here, incidentally, what really matters for the argument is the body of each rule includes an EDB atom that contains all variables occurring in the head of the rule).
Hence, the number of times a non-answer rule is applied is bounded by the 
number of facts in the input database instance times the number of rules of the Datalog program.
 Hence, the entire Datalog computation, except for the 
computation of the answer relation, can be viewed as a polynomial computation
with an NP-oracle (for evaluating the bodies of rules). Finally, once all
IDBs except the Answer IDB have been computed, we simply invoke the NP oracle 
once more to test if the given tuple belongs to the extension of the 
answer IDB.

  For the lower bound we provide a reduction from the LEX(SAT) problem: 
given a propositional formula $\Phi(x_1,\ldots,x_n)$, determine if the 
value of $x_n$ is 1 in the lexicographically least satisfying assignment, 
where $x_n$ is the least significant bit.
where $x_n$ is the least significant bit. The LEX(SAT) problem is known 
to be \PNP-complete, even for 3-CNF formulas \cite{Wagner87}.

We devise a structure $\mathfrak{B}$ with a domain of two elements ($\top$ and $\bot$) 
endowed with a unary relation $T$ that is true only of $\top$, 
a unary relation $F$ that is true only of $\bot$, a binary relation $N$ 
that holds precisely the complementary pairs $(\bot,\top)$ and $(\top,\bot)$, 
and a ternary relation $\mathrm{OR}$ 
that is true of all $\{\bot,\top\}$-triplets 
but $(\bot,\bot,\bot)$.
This way, the set of satisfying assignments to every 3-clause $C(x_i,x_j,x_k)$, 
e.g.~$x_i \lor \lnot x_j \lor \lnot x_k$, is the answer set to a corresponding 
conjunctive query $\tilde{C}(x_i,x_j,x_k)$ on $\mathfrak{B}$, such as 
$\exists y_j y_k \ N(x_j,y_j) \land N(x_k,y_k) \land \mathrm{OR}(x_i,y_j,y_k)$ in this example. 
More generally, we can translate every 3-CNF formula $\Phi(x_1,\ldots,x_n)$ 
into a Datalog rule with body 
$$
  \tilde{\Phi}(x_1,\ldots,x_n,y_1,\ldots,y_n) = \bigwedge_i N(x_i,y_i) \ \land 
                     \bigwedge_{C \text{ a clause}} \tilde{C}(\textbf{x},\textbf{y})
$$ 
where $C$ ranges over the clauses of $\Phi$. 

Given a propositional formula $\Phi(x_1,\ldots,x_n)$, the idea is now 
to have, for each $i\leq n$, a unary IDB predicate $X_i$ that computes 
the truth value of the $i$-th bit in the lexicographically least 
satisfying assignment to $\Phi(x_1,\ldots,x_n)$. 
These IDBs belong to different strata of the program as inductively 
defined by the following rules.
\[\small\begin{array}{rcl}
  Z_1      & \leftarrow & F(x_1), \, \tilde{\Phi}(\textbf{x},\textbf{y}) \\
  X_1(x_1) & \leftarrow & F(x_1), \, Z_1 \\
  X_1(x_1) & \leftarrow & T(x_1), \, \lnot Z_1 \\
           & : &  \\
  Z_i      & \leftarrow & X_1(x_1),\ldots,X_{i-1}(x_{i-1}), \, 
                          F(x_i), \, \tilde{\Phi}(\textbf{x},\textbf{y}) \\
  X_i(x_i) & \leftarrow & F(x_i), \, Z_i \\
  X_i(x_i) & \leftarrow & T(x_i), \, \lnot Z_i  \\  
           & : & \\
  Ans      & \leftarrow & X_n(x_n), \, T(x_n), \,  \tilde{\Phi}(\textbf{x},\textbf{y})
\end{array}\]
It is easy to see that the above non-recursive GN-Datalog query computes 
on $\mathfrak{B}$ the solution to the LEX(SAT) problem instance $\Phi$.
\end{proof}

\subsection{Proof of Proposition~\ref{prop:OWA_linModel}}

In what follows it will be convenient to work with GNFO formulas in 
disjunctive normal form. Two critical dimensions of a GNFO formula in DNF 
are its `width', introduced above, and its `negation rank'.
The \emph{negation rank} $\mathrm{nrank}(\phi)$ of $\phi$ in DNF is the maximum 
number of nested negations in $\phi$. %
Naturally, UCQs have negation rank $0$. Set
$
 \text{DNF}_w^r = \{ \phi \mid \phi \text{ in DNF}, \, 
                  \mathrm{width}(\phi) \leq w, \, \mathrm{nrank}(\phi) < r \}  
$.

\smallskip

Given a structure $U$ we let $\atoms(U)$ denote the set 
of tuples $\textbf{u}$ forming the support of a relational atom in $U$. 
For every $\textbf{u} \in \atoms(U)$ and $\textbf{u}' \in \atoms(U')$
and for every $w, r \in \mathbb{N}_{\geq0}$ let $U,\textbf{u} \equiv_w^r U',\textbf{u}'$ 
denote the fact that $U \models \psi(\textbf{u}) \iff U' \models \psi(\textbf{u}')$ 
for every $\psi \in \text{DNF}_w^r$. 
Assuming an ambient finite relational signature, each $\equiv_w^r$ 
is an equivalence of finite index as there are, up to logical equivalence, 
only finitely many formulas in $\text{DNF}_w^r$.
In particular, for every $\textbf{u} \in \atoms(U)$ there is a formula 
$\chi_{U,\textbf{u}}(\textbf{x})$ that is a boolean combination 
of $\text{DNF}_w^r$-formulas and is characteristic of its $\equiv_w^r$-class, 
i.e. such that $U,\textbf{u} \equiv_w^r U',\textbf{u}'$ 
iff $U' \models \chi_{U,\textbf{u}}(\textbf{u}')$. 
We call $\chi_{U,\textbf{u}}(\textbf{x})$ the 
\emph{$\mathrm{DNF}_w^r$-type of $\textbf{u}$ in $U$}. 
Note that $\chi_{U,\textbf{u}}(\textbf{x})$ is itself \emph{not} in $\text{DNF}_w^r$. 
Also note that $\text{DNF}_w^0$ is empty %
and that $\text{DNF}_w^1$ comprises only UCQs. %

\begin{proof} 
For the purposes of this construction we consider the ammendment of the language 
of $\varphi(\textbf{x})$ with constants $\textbf{c}$ corresponding to the free  
variables $\textbf{x}$, and regard $\varphi$ as the sentence obtained from the 
original query by substitution of each constant $c_i$ in place of 
the corresponding free variable $x_i$.
Accordingly, given an instance $I$ with distinguished elements $\textbf{a}$, 
we treat $\textbf{a}$ as the interpretation of the constants $\textbf{c}$. 
Furthermore, we assume w.l.o.g.\ that $\phi$ is in DNF as in~\eqref{eq:DNF} on 
page~\pageref{eq:DNF} and let $w$ be the width and $r$ the negation rank of $\phi$.  

For each $\text{DNF}_w^r$-type $\tau(\textbf{x})$ such that 
$\phi \land \tau(\textbf{x})$ is satisfiable we fix in advance, 
and independently of $I$, a finite model of $\phi$ 
with distinguished elements $(M^\tau,\textbf{a}^\tau)$ 
realising it: $M^\tau \models \phi \land \tau(\textbf{a}^\tau)$.
Let $C$ be the maximum number of facts in any of the $M^\tau$. 
Note that $C$ is independent of $I$.   

To obtain the model $M$, for every $\textbf{b} \in \atoms(I)$ having $\text{DNF}_w^r$-type 
$\tau(\textbf{z})$ in $J$ we take $(M^{\textbf{b}},\textbf{a}^{\textbf{b}})$ to be 
a fresh copy of $(M^\tau,\textbf{a}^\tau)$ and attach it to $I$ by identifying 
its distinguished tuple $\textbf{a}^{\textbf{b}}$ with $\textbf{b}$ component-wise. 
Thus $M$ is made up of at most $C n$ many facts. 
It remains to verify that $M \models \phi$. 

\begin{claim} \label{claim:identype}
For every $\textbf{b} \in \atoms(I)$ and $\textbf{d} \in \atoms(M^{\textbf{b}})$ 
we have $M^{\textbf{b}},\textbf{d} \equiv_w^r M,\textbf{d}$.
\end{claim}

From this claim it follows trivially that $M \equiv_w^r M^{\textbf{b}}$, 
therefore also $M \equiv_w^r J$, since $M^{\textbf{b}} \equiv_w^r J$ by choice.
Because $J \models \phi$, this will allow us to conclude $M \models \phi$. 

To establish Claim~\ref{claim:identype} we prove by induction on $q=0,\ldots,r$ 
that $M^{\textbf{b}},\textbf{d} \equiv_w^q M,\textbf{d}$ for every $\textbf{b} \in \atoms(I)$ 
and $\textbf{d} \in \atoms(M^{\textbf{b}})$. The latter claim is trivially true for $q=0$. 
Towards the induction step assume it is true for $q-1$ and 
consider an arbitrary $\textbf{b} \in \atoms(I)$ and $\textbf{d} \in \atoms(M^{\textbf{b}})$. 
It suffices to show that 
  $M^{\textbf{b}} \models \psi(\textbf{d}) \iff M \models \psi(\textbf{d})$ 
for all %
$$
  \psi(\textbf{x}) = \exists \textbf{y} \ 
      \bigwedge_l ( \alpha_l(\textbf{z}^l) \land \lnot \psi_l(\textbf{z}^l) ) \\[-.7em]
$$
where each $\alpha_l(\textbf{z}^l)$ is an atomic formula and $\psi_l\in\text{DNF}_w^{q-1}$   
with free variables $\textbf{z}^l$ from among $\textbf{x}\textbf{y}$. 
For $\psi$ as above we additionally define 
$\nu(\textbf{x},\textbf{y}) = \bigwedge_l ( 
       \alpha_l(\textbf{z}^l) \land \lnot \psi_l(\textbf{z}^l) )$.

  Tackling first the easy direction, suppose that $M^{\textbf{b}} \models \psi(\textbf{d})$  
and consider witnesses $\textbf{c}$ in $M^{\textbf{b}}$ such that 
$M^{\textbf{b}} \models \nu(\textbf{d},\textbf{c})$. 
Then $M \models \alpha_l(\textbf{e}^l)$ is immediate for each subtuple $\textbf{e}^l$ 
that relates to $\textbf{d}\textbf{c}$ as $\textbf{z}^l$ relates to $\textbf{x}\textbf{y}$, 
while $M \models \lnot \psi_l(\textbf{e}^l)$ follows from 
$M^{\textbf{b}} \models \lnot \psi_l(\textbf{e}^l)$ via the induction hypothesis.
This proves $M \models \psi(\textbf{d})$.

  Suppose now $M \models \psi(\textbf{d})$ and let $\textbf{c}$ be elements of $M$ 
such that $M \models \nu(\textbf{d},\textbf{c})$. 
Our aim is to find witnesses $\textbf{c}'$ in $M^{\textbf{b}}$ 
such that $M^{\textbf{b}} \models \nu(\textbf{d},\textbf{c}')$. 
We distinguish two cases.
\newline\noindent
(i) If $\textbf{c}$ lies entirely in $M^{\textbf{b}}$ then, 
   using the induction hypothesis as in the proof of the opposite direction, 
   we can confirm that $\textbf{c}'=\textbf{c}$ are appropriate witnesses: 
   $M^{\textbf{b}} \models \nu(\textbf{d},\textbf{c})$.
\newline\noindent
(ii) Otherwise we proceed as follows. 
For each atom $\alpha_l(\textbf{z}^l)$ from $\nu$ let $\textbf{e}^l$ be the subtuple 
relating to $\textbf{d}\textbf{c}$ as $\textbf{z}^l$ relates to $\textbf{x}\textbf{y}$. 
Thus $M \models \alpha_l(\textbf{e}^l) \land \lnot \psi_l(\textbf{e}^l)$ for each $l$. 

Next, for each $\textbf{a} \in \atoms(I)$ let $\lambda(\textbf{a})$ be the set 
of those indices $l$ such that $\textbf{e}^l$ lies entirely in $M^{\textbf{a}}$ 
and let $\textbf{e}^{(\textbf{a})}$ enumerate (without repetition) all elements 
from those $\textbf{e}^l$ with $l \in \lambda(\textbf{a})$.
Let in addition $\delta^{\textbf{a}}$ be the conjunction of all those formal 
equalities $a_j=\big(\textbf{e}^{(\textbf{a})}\big)_k$ that hold in $M$.
Further let $\nu^{\textbf{a}} = \delta^{\textbf{a}} \land \, 
  \bigwedge_{l\in\lambda(\textbf{a})} \alpha_l(\textbf{e}^l) \land \lnot \psi_l(\textbf{e}^l)$
and $\psi^{\textbf{a}} = \exists \textbf{e}^{(\textbf{a})} \nu^{\textbf{a}}$. 

Thus, $M \models \nu^{\textbf{a}}(\textbf{a},\textbf{e}^{(\textbf{a})})$ and we find, as in (i), 
that also $M^{\textbf{a}} \models \nu^{\textbf{a}}(\textbf{a},\textbf{e}^{(\textbf{a})})$, 
hence $M^{\textbf{a}} \models \psi^{\textbf{a}}(\textbf{a})$ and, 
because $M^{\textbf{a}},\textbf{a} \equiv_w^r J,\textbf{a}$, we also learn 
that $J \models \psi^{\textbf{a}}(\textbf{a})$. 
So there are $\textbf{u}^{(\textbf{a})}$ in $J$ 
such that $J \models \nu^{\textbf{a}}(\textbf{a},\textbf{u}^{(\textbf{a})})$.
Let $\textbf{u}$ enumerate all $\textbf{u}^{(\textbf{a})}$ 
for $\textbf{a} \in \atoms(I)$ different from $\textbf{b}$. 
Note that, crucially, $|\textbf{u}| \leq w$. 
Indeed, because for each $\textbf{a}$ elements of $\textbf{a}$ and $\textbf{u}^{(\textbf{a})}$ 
satisfy the equalities prescribed in $\delta^{\textbf{a}}$, 
it is ensured that 
any equalities between elements of witnessing tuples $\textbf{e}^{(\textbf{a})}$ 
in $M^{\textbf{a}}$ and $\textbf{e}^{(\textbf{a}')}$ in $M^{\textbf{a}'}$ 
with $\textbf{a}\neq\textbf{a}'$ (which, by definition of $M$, must necessarily 
involve elements occurring both in $\textbf{a}$ and in $\textbf{a}'$) 
are also observed by the corresponding witnesses $\textbf{u}^{(\textbf{a})}$ 
and $\textbf{u}^{(\textbf{a}')}$ in $J$. 

Altogether we have 
$J \models \bigwedge_{\textbf{a}\neq\textbf{b}} 
               \nu^{\textbf{a}}(\textbf{a}, \textbf{u}^{(\textbf{a})})$,
where $\textbf{a}$ ranges over $\atoms(I) \setminus \{\textbf{b}\}$.
Let $\delta^\textbf{b}$ be the conjunction of all equalities $b_j=u_k$ 
that do hold in $J$, and let 
$
  \zeta(\textbf{b},\textbf{u}) = \delta^\textbf{b} \land \, 
      \bigwedge_{\textbf{a}\neq\textbf{b}} \bigwedge_{l\in\lambda(\textbf{a})} 
       \alpha_l(\textbf{e}^l) \land \lnot \psi_l(\textbf{e}^l)      
$,
where $\textbf{a}$ ranges over $\atoms(I) \setminus \{\textbf{b}\}$.
By the above, $J \models \exists \textbf{u} \, \zeta(\textbf{b},\textbf{u})$ 
and $\exists \textbf{u} \, \zeta(\textbf{b},\textbf{u}) \in \text{DNF}_w^q$,  
so from $M^{\textbf{b}},\textbf{b} \equiv_w^r J,\textbf{b}$ it follows that 
 $M^{\textbf{b}} \models \exists \textbf{u} \, \zeta(\textbf{b},\textbf{u})$.
Taking into account that $\delta$ stipulates all equalities between components 
of $\textbf{b}$ and any witnesses $\textbf{u}$ to $\zeta$ in $M^{\textbf{b}}$ that are 
valid in $J$ and, correspondingly, that are valid in $M$ between elements 
of $\textbf{b}$ and the original witnesses $\textbf{c}$ to $\psi$ in $M$, we can conclude 
from this and from of course $M^{\textbf{b}} \models \psi^{\textbf{b}}(\textbf{d})$ 
that $M^{\textbf{b}} \models \psi(\textbf{d})$ as needed.
This completes the induction step in the proof of Claim~\ref{claim:identype}.
\end{proof}

\subsection{Proof of Theorem~\ref{thrm:OWA_SGNQ_PTIME}}

In the proof below, we concentrate on boolean queries. The PTime 
upper bound for boolean queries extends immediately to non-boolean queries:
consider a $k$-ary SGNQ $\phi(x_1, \ldots, x_k)$. Let $P_1, \ldots, P_n$
be fresh monadic predicates. Then the problem of testing whether 
$I\models_{OWA} \phi(a_1, \ldots, a_k)$, for given $I$ and $a_1, \ldots, a_k$, reduces
to the problem whether 
$I'\models_{OWA}\exists x_1, \ldots, x_k(\phi(x_1,\ldots,x_k)\land P_1(x_1)\land\cdots\land P_k(x_k))$,
where $I'$ extends $I$ by interpreting each new predicate $P_i$ by the singleton set $\{a_i\}$. 

\begin{proof}
Consider a boolean SGNQ $Q$. By prescription, conjunctions under 
an even number of negations in $Q$ may contain at most one conjunct 
that is a negated subformula. To allow for a uniform treatment we introduce 
a fresh nullary predicate $\mathtt{false}$ and add $\lnot \mathtt{false}$ 
as a new conjunct in those subformulas of $Q$ (whether in the context of 
an even or odd number of negations), where there were no negative conjuncts.
After this trivial transformation $Q$ takes the form
\begin{equation} \label{eq:disjfree}
  \mathtt{false} \ \lor \ 
  \bigvee_i \exists \textbf{x} \left( 
      \alpha_i(\textbf{x}) \land \lnot \exists \textbf{y} \ \psi_i(\textbf{x},\textbf{y}) 
  \right)  
\end{equation}
where each $\alpha_i$ is a conjunction of atoms and $\psi_i$ is 
a DNF fromula built with only $\exists$, $\land$ and guarded negation.  
We allow above $|\textbf{x}|=0$ and $\alpha_i$ to be an empty conjunction, 
i.e. vacuously true. Thus, \eqref{eq:disjfree} contains, as a special case, 
disjuncts of the form $\lnot \exists \textbf{y} \ \psi(\textbf{y})$. 
As another special case, \eqref{eq:disjfree} may contain disjuncts 
$\exists \textbf{x} \ ( \alpha_i(\textbf{x}) \land \lnot \mathtt{false} )$
with $\psi_i = \mathtt{false}$ and the corresponding quantification 
$\exists \textbf{y}$ being vacuous ($|\textbf{y}|=0$). 
We shall write $Q$ equivalently as 
\begin{equation} \label{eq:disjfreeneg}
  \bigwedge_i \forall \textbf{x} \left( 
      \alpha_i(\textbf{x}) \limp \exists \textbf{y} \ \psi_i(\textbf{x},\textbf{y}) 
  \right) 
  \ \limp \ \mathtt{false} 
\end{equation}
akin to a formulation of CQ entailment of frontier-guarded tgds -- except for the fact 
that $\psi_i(\textbf{x},\textbf{y})$ need not be quantifier free. 
Via induction on the quantifier alternation rank we show that~\eqref{eq:disjfreeneg} 
can be `flattened' to an equi-satisfiable $\forall\exists$-formula of GNFO 
asserting that a conjunction of frontier-guarded tgds entails $\mathtt{false}$. 

Each $\psi_i$ is in disjunctive normalform 
and occurs in the scope of an odd number of negations in the 
disjunction-free $\phi$. Hence it assumes the following general form 
$$
  \psi_i(\textbf{x},\textbf{y}) = 
  \gamma_i \ \land \ \bigwedge_{m} \lnot \delta_{i,m} \ \land \  
  \bigwedge_{n} \lnot \exists \textbf{z} \ (\gamma'_{i,n} \land \lnot \xi_{i,n})  
$$
where $\gamma_i$ and $\gamma'_{i,n}$ are conjunctions of positive atoms,
each $\delta_{i,m}$ is an atom and each $\xi_{i,m}$ is either an atom, 
or an existentially quantified formula, or $\mathtt{false}$.  
Let $W_i(\textbf{x},\textbf{y})$ be a new predicate symbol of the same arity as $\psi_i$. 
We replace in \eqref{eq:disjfreeneg} the $i$-th conjunct with a collection of 
new conjuncts:   
\[\begin{array}{ll}
  \forall \textbf{x} \ \left( \alpha_i(\textbf{x})  
      \limp \exists \textbf{y} \ W_i(\textbf{x},\textbf{y}) \right) 
      & \\
  \forall \textbf{x} \textbf{y} \ \left( W_i(\textbf{x},\textbf{y}) 
      \limp \gamma_i  \right)    
      & \\
  \forall \textbf{x} \textbf{y} \ \left( W_i(\textbf{x},\textbf{y}) \land \delta_{i,m} 
      \limp \mathtt{false} \right) 
      & \text{for each } \delta_{i,m}   \\
  \forall \textbf{x} \textbf{y} \textbf{z} \ \left( 
      W_i(\textbf{x},\textbf{y}) \land \gamma'_{i,n} \limp \xi_{i,n} \right) 
      & \text{for each } \gamma'_{i,n}    
\end{array}\]
Because all negations were properly guarded, all of these rules 
are frontier-guarded tgds, save perhaps some of those of the last kind with $\xi$ 
an existentially quantified formula, which then pertain to the same 
restrictions as the original conjuncts of \eqref{eq:disjfreeneg},  
only having a lower quantifier alternation rank. 
Iterating this transformation one eventually arrives at the desired 
form comprising only frontier-guarded tgds.   
 
The theorem follows from the fact that OWA query answering against 
frontier-guarded tgds has PTime data complexity \cite{Baget11ijcai}. 
In fact, \cite{Baget11rep} shows that every CQ can be rewritten relative to 
a set of frontier-guarded tgds into a Datalog program that can be executed on a 
database instance to yield the OWA answer to the original query. 
Thus, each serial GNFO query $Q$ can also be reformulated as a
Datalog program $(\Pi, \mathtt{false})$ such that  
$I \models_{\mathrm{OWA}} Q \, \iff \, \Pi(I) \models \mathtt{false}$ 
for all instances $I$.    
\end{proof}

\subsection{Proof of Theorem~\ref{thrm:OWA_SGNQineq_undec}}

\begin{proof}
We combine ideas of \cite{Graedel99JSL} and \cite[Theorem 15]{CGK08kr} to encode 
computations of a fixed Turing machine with the database instances 
representing input words. 
Using guarded tgds and an egd (in fact, a key constraint), we will force the existence of a grid frame 
onto which a valid computation of the Turing machine is charted. 
The advantage of using tgds and egds is the well-known fundamental 
principle \cite{JK84,FKMP05,NDR06rep} that open-world query answering 
on an instance $D$ relative to a set $\Sigma$ of tgds and egds reduces 
to query evaluation on the single universal (though generally infinite) 
chase model $chase(D,\Sigma)$. 
While the chase model with respect to guarded tgds is always tree-like 
\cite{CGK08kr,Baget11ijcai}, the additional key constraint 
imposed by the egd enforces a grid-like structure of the chase model.

Let $M$ be a Turing machine, to be chosen later, having tape alphabet $A$, 
states $Q$ and transition function $\delta: Q\times A \to Q \times A \times \{-1,0,1\}$.
Input words to $M$ will be presented as successor-structures: comprising 
a $succ$-chain of $A$-labelled elements. The signature consists of a binary 
relation $succ$ and unary relations $P_a$ for every $a \in A$. 
We assume that the predicates $P_a$ partition the input structure, 
that $A$ contains a special start symbol $\triangleright$ labelling 
only the first element and a special blank symbol $\flat$ labelling 
only the last element of the successor-chain.  
 
Next we define a set $\Sigma_M$ of guarded tgds and egds over an expanded 
signature responsible for simulating $M$. $\Sigma_M$ has as conjuncts the 
following guarded tgds (omitting the implicit universal quantification 
of variables in rule bodies). 
\[\begin{array}{l}
  succ(x,y) \limp \exists z \ succ(y,z) \qquad 
  succ(x,y) \land P_\flat(x) \limp P_\flat(y) \\
  succ(x,y) \limp \exists u v \ cell(x,y,u,v) \\
  cell(x,y,u,v) \limp next(x,u) \land next(y,v) \land succ(u,v) \\
\end{array}\]
In addition, $\Sigma_M$ contains the key constraint 
\begin{equation} \label{eq:egds}
\begin{array}{c}
  next(x,y) \land next(x,z) \limp y=z
\end{array}
\end{equation}
expressing functionality of $next$. %
It is easy to see that the infinite chase of any input structure 
as specified above wrt. these guarded tgds and the egd is an 
infinite grid with $succ$ and $next$ acting as horizontal and 
vertical successor edges and whose bottom $succ$-chain is labelled 
with $\triangleright w \flat^\omega$, where $w$ is the input word.
Consequently, every model of these rules embeds a homomorphic 
image of this grid. 

The next step is to implement, given the grid frame, the workings of 
the Turing machine $M$ using additional guarded tgd rules. 
To this end the we will make use of additional unary predicates $S_q$ 
for every state $q \in Q$ of $M$. Let $init$ be the initial state 
and $acc$ the w.l.o.g.\ unique accepting state of $M$. 
To initiate the computation $\Sigma_M$ specifies 
$$
   P_\triangleright(x) \limp S_{init}(x)
$$ 
and to carry it on $\Sigma_M$ contains guarded tgds associated to 
each transition $(p,a,q,b,\iota) \in \delta$.
\[\begin{array}{l} 
  cell(x,y,u,v) \land S_p(x) \land P_a(x) \limp P_b(x) \land S_q(u)
  \ \text{ for } \iota=0 \\
  cell(x,y,u,v) \land S_p(x) \land P_a(x) \limp P_b(x) \land S_q(v)
  \ \text{ for } \iota=1 \\
  cell(x,y,u,v) \land S_p(y) \land P_a(y) \limp P_b(y) \land S_q(u)
  \ \text{ for } \iota=-1 \\
\end{array}\]
To ensure that tape symbols not affected by a transition are copied 
from one configuration to the next we add unary predicates $L$ and $R$ 
(intuitively, $L$ and $R$ mark the positions left and right of the head 
of a configuration, respectively) and the following guarded tgd rules. 
\[\small\begin{array}{c} 
  succ(x,y) \land S_q(x) \limp R(y) \qquad 
  succ(x,y) \land R(x) \limp R(y) \\[0.2em]
  succ(x,y) \land S_q(y) \limp L(x) \qquad
  succ(x,y) \land L(y) \limp L(x) \\[0.3em]
  cell(x,y,u,v) \land R(y) \land P_a(y) \limp P_a(v) \\[0.2em] 
  cell(x,y,u,v) \land L(x) \land P_a(x) \limp P_a(u) \\
\end{array}\]
This completes the specification of the set $\Sigma_M$ of guarded tgds 
and the single egd responsible for simulating the Turing machine $M$. 
It should be clear that $M$ accepts a word $w$ if, and only if, 
the corresponding instance $D_w$ satisfies 
$$
  D_w, \Sigma_M \ \models_\mathrm{OWA} \ \exists x \ S_{acc}(x) \ .
$$
The first claim of the theorem now follows by choice of some Turing machine $M$
that accepts an r.e.-complete language. 
For the second claim consider the key constraint~\eqref{eq:egds} 
and the GNFO query $\varphi_M \lor \exists x \ S_{acc}(x)$, where $\varphi_M$ 
is the disjunction of the negations of the guarded tgds of $\Sigma_M$.   
\end{proof}

\subsection{Proof of Proposition~\ref{fmpobs}}

\prf
For the equivalence between (v) and (vi) we merely note that 
the stage increments $X^{n+1} \setminus X^n$  for each IDB predicate $X$
are $\GNF$-definable, for each $n \in \mathbb{N}$. Now 
$\Pi$ is classically unbounded 
if, and only if, for at least one $X$,
these formulas are individually satisfiable, 
for every $n \in \mathbb{N}$; and similarly in restriction to
finite instances. The finite model property for $\GNF$ 
therefore shows the equivalence. 

(ii) $\Rightarrow$ (iv) and (i) $\Rightarrow$ (iii)
follow from the characterizations of $\GNF$ as a fragment of
$\FO$ in terms of preservation under suitable notions of guarded negation 
bisimulation ($w$-bounded guarded negation 
bisimulation), as presented in~\cite{BtCS11} for the classical version 
and in~\cite{O2012} for the finite model theory version. 
These apply since all stages of $\Pi$ 
(finite and infinite, if we admit infinite instances) and especially 
the limit $\Pi^\infty$ are preserved under $w$-bounded guarded negation 
bisimulation, if $\Pi$ is of width $w$. 

(iv) $\Rightarrow$ (vi) is the natural variant  
of the classical Barwise--Moschovakis theorem for $\GNF$, which 
may be obtained from the classical via the semantic 
characterisation of $\GNF$  as a fragment of $\FO$ in~\cite{BtCS11}. 

We concentrate on (iii) $\Rightarrow$ (iv).
Assume that formulas $\psi_X(\xbar_i) \in \GNF$ define $X^\infty$
across all finite instances, but fail to define $\Xbar^\infty = \Pi^\infty(I)$ 
over some infinite instance $I$. Appealing to the form of the rules in
$\Pi$, we assume w.l.o.g.\ that $\psi_X(\xbar)$ is 
explicitly guarded in the form 
$\psi_X(\xbar) = 
\bigvee_s \bigl( \alpha_{s}(\xbar_{s}) \wedge
\psi_X(\rho_{s}(\xbar))\bigr)$, 
where every rule in $\Pi$ with head predicate $X$ 
gives rise to one disjunct, and $\rho_{s}$ is the appropriate 
substitution to match the variable tuple $\xbar$ onto the $\xbar_{s}$
used in that rule (in particular, $\alpha_i$ guards
all free variables in $\psi_X(\rho_{s}(\xbar))$).

The fact that a tuple of predicates $\Pbar$ 
is a fixed point of $\Pi$ is expressible by a sentence 
$\chi\in \GNF$ 
(in the signature extended with new $P_X$, one for each IDB predicate 
$X$ in $\Xbar$, which may even be used as
guards). But for the tuple or predicates defined by the $\psi_X$, there is
even a sentence $\xi \in \GNF$ in the basic (EDB) signature 
saying that this tuple is a fixed point of $\Pi$: 
the crucial point to note is that 
these predicate equalities reduce to set inclusions under 
each one of the relevant guards $\alpha_s$ (!). 
If one of the $\psi_X$ failed over any infinite instance, then,
by the finite model property for $\GNF$, it would also fail over some finite
instance. So the $\psi_X$ must define a fixed point of $\Pi$ across all, 
finite and infinite, instances. A similar argument shows that 
this fixed point defined by the $\psi_X$ over an 
infinite instance $I$ must be the least fixed point $\Xbar^\infty$.
Otherwise, there would have to be some other, strictly smaller fixed point
$\Pbar$ (viz.\ $\Pbar := \Xbar^\infty$). This fact can also 
be expressed by a sentence of $\GNF$ in the signature extended by
the new predicate letters $P$. So the finite model property for
$\GNF$ would again pull this situation down to some finite instance --
contradicting the assumption that the $\psi_X$ define $\Xbar^\infty$
over all finite instances. 
\eprf

\subsection{Proof of Theorem~\ref{thm:inequality-undecidable}}

\begin{proof}
  It is known that the implication problem for inclusion dependencies and key constraints
      lacks finite controllability, and is undecidable both on finite and on unrestricted instances (cf.~\cite{AHV95}).
      It follows that also the satisfiability and query containment problems for GN-SQL are not finitely controllable, and are undecidable both on finite instances and on unrestricted instances.
    As for the last item, it follows from Theorem~\ref{thrm:OWA_SGNQineq_undec}, using the fact that key constraints (being a special case of functional dependencies) can be expressed in GN-SQL($\neq$).
     Specifically, the GN-SQL(${\neq}$) query for which open world query answering is undecidable, is $q_1 \textsf{ union } q_2$ where $q_1$ is the boolean GN-SQL query from Theorem~\ref{thrm:OWA_SGNQineq_undec}(ii)
    and $q_2$ is the boolean GN-SQL($\neq$) query expressing the negation of the key constraint from Theorem~\ref{thrm:OWA_SGNQineq_undec}(ii).
\end{proof}

\subsection{Proof of Theorem~\ref{thm:ordering-containment}}

\begin{proof}[(sketch)]
  In translating GN-SQL(\textsc{lin}) to GNFO, we have to overcome a discrepancy in the use of constants. The constants that may appear in a GN-SQL(\textsc{lin}) query are actual values from the linearly ordered domain \textsc{lin}. GNFO, on the other hand, allows for the use of constant symbols, whose interpretation is given by the structure, and may differ between structures. In particular, a structure may interpret two constant symbols by the same element. In order to overcome this discrepancy, we (i) introduce for each element $d$ of $\textsc{lin}$ a corresponding constant symbol \textsf{d}, and (ii) we construct a GNFO sentence that ``axiomatizes'' the correct behavior of the constant symbols (including the fact that distinct constant symbols denote different values). 

\newcommand{\xor}{~\underline{\lor}~}
More precisely, let $\textsc{lin}=(D,\prec)$ and for each finite subset $S=\{d_1, \ldots, d_n\}$ of $D$, with $d_1\prec\ldots\prec d_n$, let $\theta_S$ be the following GNFO sentence, containing a constant symbol $\mathsf{d}_i$ for each $d_i\in S$:
$$\forall x \Big(\phi_{x<\mathsf{d}_1}\lor \phi_{x=\mathsf{d}_1}\lor \phi_{\mathsf{d}_1<x<\mathsf{d}_2}\lor \phi_{x=\mathsf{d}_2} \lor\cdots\lor\phi_{x>\mathsf{d}_n}\Big)$$ 
where
\[\small \hspace{-2mm}\begin{array}{@{}l@{~}l} 
  \phi_{x<\mathsf{d}_1} = & \begin{cases} 
      \bigwedge_{i\leq n} (x<\mathsf{d}_i \land \neg(x=\mathsf{d}_i)\land \neg(\mathsf{d}_i<x)) & \text{if } \exists d\in D ~ d\prec d_1 \\ 
      \bot & \text{otherwise} \end{cases} \\[1em]
  \phi_{x=\mathsf{d}_i} = & (x=\mathsf{d}_i)\land \neg(x<\mathsf{d}_i)\land \neg(\mathsf{d}_i<x)  \\& 
     ~~\land \bigwedge_{j<i} (\mathsf{d}_j<x \land \neg(\mathsf{d}_j=x)\land \neg(x<\mathsf{d}_j)) \\&
     ~~\land \bigwedge_{j>i} (x<\mathsf{d}_j \land \neg(\mathsf{d}_j=x)\land \neg(\mathsf{d}_j<x)) \\[1em]
  \phi_{\mathsf{d}_i<x<\mathsf{d}_{i+1}} = \hspace{-7mm} & \hspace{7mm}  \begin{cases} 
      \bigwedge_{j\leq i}(\mathsf{d}_j<x\land \neg(\mathsf{d}_j= x) \land \neg(x<\mathsf{d}_j)) \\
      ~\land \bigwedge_{j> i}(x<\mathsf{d}_j\land \neg(\mathsf{d}_i= x) \land \neg(\mathsf{d}_j<x))) \\
     & \hspace{-25mm} \text{if } \exists d\in D ~ d_i\prec d\prec d_{i+1} \\ 
      \bot & \hspace{-25mm}\text{otherwise} \end{cases} \\[1em]
  \phi_{\mathsf{d}_n<x} = & \begin{cases} 
      \bigwedge_{i\leq n} (\mathsf{d}_i<x\land \neg(x=\mathsf{d}_i)\land \neg(x<\mathsf{d}_i)) & \text{if } \exists d\in D ~ d_n\prec d \\ 
      \bot & \text{otherwise} \end{cases} \\[1em]
\end{array}\]

Observe that this set $\theta_S$ can be constructed from $S$ in polynomial time, since $\textsc{lin}$ is reasonable.  Furthermore, the following crucial property holds: if $M$ is any structure satisfying $\theta_S$, and if $M'$ is an isomorphic copy of $M$ in which each constant symbol $\textsf{d}_i$ denotes the actual corresponding value $d_i$ (for all $d_i\in S$), then $M$ and $M'$ are indistinguishable with respect to GN-SQL(\textsc{lin}) queries whose constants are included in $S$.
It follows that a GN-SQL(\textsc{lin})  query $q_1$ is contained in a GN-SQL(\textsc{lin}) query $q_2$ if and only if, for their GNFO translations $q^*_1$ and $q^*_2$, we have that
$q_1^*\land \theta_S\models q_2^*$, where $S$ is the set of constants occurring in $q_1$ and $q_2$. 
\end{proof}

\subsection{Proof of Proposition~\ref{prop:IDB-guards}}

\begin{proof}
The upper bound follows from the ExpTime upper bound for stratified Datalog \cite{DEGV01}
(obtained by the standard technique of ``grounding'' a program by instantiating 
its rules via substituting domain elements for the variables in every possible 
way and solving the resulting exponentially large propositional Horn program 
with stratified negation by standard means).  

For the lower bound we provide a reduction from the acceptance problem for 
polynomial-space alternating Turing machines. 
Consider an alternating Turing machine $M$ using $p(n)$ tape cells on 
any input of length $n$. We may assume w.l.o.g.\ that the states of $M$ 
are partitioned into existential states $Q_\exists$ and universal states $Q_\forall$ 
and that the transition table of $M$ 
consists 
of tuples $(p,a,q,b,\epsilon,s,c,\delta)$ 
interpreted as follows. When in state $p \in Q$ and reading $a$ there are two possible 
transitions: writing $b$ at the current position, entering state $q$, and moving the 
read-write head by $\epsilon \in \{-1,0,+1\}$; or writing $c$, entering state $s$ and 
moving the head by $\delta \in \{-1,0,+1\}$. The choice between the two possibilities 
is existential or universal according to whether $p\in Q_\exists$ or $p \in Q_\forall$.
In addition we may assume w.l.o.g.\ that $M$ has a unique initial state $init$, 
and a unique accepting configuration with state $acc$, head position $1$ and 
the used segment of its tape filled with $0$'s. 

Let $\mathfrak{B}$ be the structure with domain $\{0,1\}$ and %
the binary relation $Bits$ that holds the pair $(0,1)$ alone. 
Given an input word $w$ of length $|w|=n$ and $M$ as above we let $N=p(n)$ and 
we devise a GN-Datalog program with IDB predicates $S_{q,i}(u_1,\ldots,u_N,z,o)$ 
of arity $N+2$ for each $q \in Q_\exists \cup Q_\forall$ and $1\leq i \leq N$.
Intuitively speaking, every fact $S_{q,i}(u_1,\ldots,u_N,z,o)$ will encode 
a configuration of $M$ in state $q$, head location $i$ and tape contents $u_1\ldots u_N$,
and $z$ and $o$ will invariably contain the values $0$ and $1$ in every 
such fact ever derived.
The GN-Datalog program $\Pi_{M,n}$ simulating $M$ on an $N=p(n)$-bounded tape comprises 
the following rules. For every transition $(p,a_0,q,a_1,\epsilon,s,a_2,\delta)$ 
and every $i$ such that $ 1 \leq i,i+\epsilon,i+\delta \leq N$, 
if $p \in Q_\exists$ then there are rules
\[\small\begin{array}{l}
  S_{p,i}(u_1,\ldots,u_{i-1},\sigma_0,u_{i+1},\ldots,u_N,z,o) \\ 
     \qquad \leftarrow \  
      S_{q,i+\epsilon}(u_1,\ldots,u_{i-1},\sigma_1,u_{i+1},\ldots,u_N,z,o) \, . \\
  S_{p,i}(u_1,\ldots,u_{i-1},\sigma_0,u_{i+1},\ldots,u_N,z,o) \\ 
     \qquad \leftarrow \    
      S_{s,i+\delta}(u_1,\ldots,u_{i-1},\sigma_2,u_{i+1},\ldots,u_N,z,o) \, . \\
\end{array}\]
and if $p \in Q_\forall$ then there is a rule
\[\small\begin{array}{l}
  S_{p,i}(u_1,\ldots,u_{i-1},\sigma_0,u_{i+1},\ldots,u_N,z,o) \\
     \qquad \leftarrow \  
      S_{q,i+\epsilon}(u_1,\ldots,u_{i-1},\sigma_1,u_{i+1},\ldots,u_N,z,o), \, \\
     \qquad \phantom{\leftarrow} \  
      S_{s,i+\delta}(u_1,\ldots,u_{i-1},\sigma_2,u_{i+1},\ldots,u_N,z,o) \, . \\
\end{array}\]
where in both cases 
$\sigma_j$ is $\left\{\begin{array}{ll}
               z & \text{if } a_j=0\\ 
               o & \text{if } a_j=1
            \end{array}\right.$ 
for each $j=0,1,2$.
In addition there is an acceptance rule 
\[\small\begin{array}{rcl}
  S_{acc,1}(\underbrace{z,\ldots,z}_{N\,\mathrm{times}},z,o) & \leftarrow & Bits(z,o)
\end{array}\]
corresponding to the unique accepting configuration, 
and the answer rule
\[\small\begin{array}{rcl}
  Ans_w      & \leftarrow & S_{init,1}(u_1,\ldots,u_N,z,o), Bits(z,o)  
\end{array}\]
encoding the initial configuration for a given input word $w \in \{0,1\}^n$,
where each $u_i$ is one of the variables $z$ or $o$ according to whether 
the $i$th bit of the initial tape contents with input $w$ is zero or one. 
Then $w$ is accepted by $M$ if the GN-Datalog query $(\Pi_{M,|w|},Ans_w)$ 
evaluates to true on $\mathfrak{B}$.
\end{proof}

\end{document}